\documentclass[12pt,draftcls,onecolumn]{IEEEtran}

\usepackage{setspace,multirow,tikz,etoolbox,multicol}%slashbox
\usetikzlibrary{arrows,automata}
\usepackage{lipsum, enumitem}

\usepackage{blkarray}
\usepackage[T1]{fontenc}
\usepackage[latin9]{inputenc}
\usepackage{amsthm}
\usepackage{amsmath}
\usepackage{graphicx}
\usepackage{amssymb}
\usepackage{url}
\usepackage{verbatim}
\usepackage[export]{adjustbox}
\usepackage{framed}
\usepackage[font=small]{caption}
\usepackage{multirow}
\usepackage[ruled,vlined]{algorithm2e}

\theoremstyle{plain}
\newtheorem{thm}{Theorem}
\theoremstyle{plain}

\usepackage{cite}\usepackage{amsfonts}\usepackage{times}\usepackage{bm}%\usepackage{dsfont}
\usepackage{amsthm}
\usepackage{subfigure}
\theoremstyle{definition}
\newtheorem{assumption}{Assumption}

\newtheorem{definition}{Definition}

\newtheorem{lemma}{Lemma}

\newtheorem*{lemma*}{Lemma}

\theoremstyle{remark}
\newtheorem{remark}{Remark}

\newcommand{\nm}[1]{\textcolor{blue}{\textbf{[NM: #1]}}}

\title{Multi-Scale Spectrum Sensing
in\\ Dense Multi-Cell Cognitive Networks}
\author{Nicolo Michelusi, Matthew Nokleby, Urbashi Mitra, and Robert Calderbank
\thanks{Part of this work appeared
at Globecom 2015 \cite{MichelusiGCOM},
 ICC 2017 \cite{MicheICC} and Asilomar 2017 \cite{MicheAsilomar}.}
\thanks{The research of U. Mitra has been funded in part by the following grants: ONR N00014-15-1-2550, ONR N00014-09-1-0700, NSF CNS-1213128, NSF CCF-1410009, AFOSR FA9550-12-1-0215, NSF CPS-1446901,
the Royal Academy of Engineering, the Fulbright Foundation and the Leverhulme Trust.}
\thanks{The research of N. Michelusi  has been funded by NSF under grant CNS-1642982, and by DARPA under grant \#108818.}
\thanks{N. Michelusi is with the School of Electrical and Computer Engineering, Purdue University. email: michelus@purdue.edu. M. Nokleby is with the Dept. of Electrical and Computer Engineering, Wayne State University. email:matthew.nokleby@wayne.edu. U. Mitra is with the Dept. of Electrical Engineering, University of Southern California. email: ubli@usc.edu. R. Calderbank is with the Dept. of Electrical Engineering, Duke University. email: robert.calderbank@duke.edu.}
}

\graphicspath{%
     {./Figures/}
}

\begin{document}
\maketitle
%\pagenumbering{gobble}
\begin{abstract}
Multi-scale spectrum sensing is proposed to overcome the cost of full network state information on the spectrum occupancy of primary users (PUs) in dense multi-cell cognitive networks. Secondary users (SUs) estimate the local spectrum occupancies and aggregate them hierarchically to 
estimate spectrum occupancy at multiple spatial scales. Thus, SUs obtain fine-grained estimates of spectrum occupancies of nearby cells, more relevant to scheduling tasks, and coarse-grained estimates of those of distant cells. 
 An \emph{agglomerative clustering} algorithm is proposed to design a cost-effective aggregation tree, matched to the structure of interference, robust to
 \emph{local estimation errors} and \emph{delays}.
Given these multi-scale estimates, the SU traffic
is adapted in a decentralized fashion  in each cell, to optimize the trade-off among SU cell throughput, interference caused to PUs, and mutual SU interference.
Numerical evaluations demonstrate
a small degradation in SU cell throughput (up to 15\% for a 0dB interference-to-noise ratio experienced at PUs)
 compared to a scheme with full network state information,
using only one-third of the cost incurred in the exchange of spectrum estimates.
The proposed \emph{interference-matched} design is shown to significantly outperform a random tree design, by providing more relevant information for network control,
and a state-of-the-art consensus-based algorithm, which does not leverage the
spatio-temporal structure of interference across the network.
\end{abstract}
\section{Introduction}
The recent proliferation of mobile devices has been exponential in number as well as heterogeneity \cite{CISCO}, demanding new tools for the design of agile wireless networks \cite{pcast}. Fifth-generation (5G) cellular systems will meet this challenge in part by deploying {\em dense, heterogeneous} networks, which must flexibly adapt to time-varying network conditions. Cognitive radios \cite{Mitola} have the potential to improve spectral efficiency by enabling secondary users (SUs) to exploit resource gaps left by legacy primary users (PUs)~\cite{Peha}. However, estimating these resource gaps in real-time becomes increasingly challenging 
with the increasing network densification, due to the signaling overhead required to learn the network state \cite{Wu2012}. Furthermore,
network densification results in irregular network topologies. These features  demand effective interference management to fully leverage spatio-temporal spectrum access opportunities.

%On the other hand, there has been increasing interest in the research community in developing systems utilizing frequencies in the 28-100 GHz range, the so called 
% \emph{millimeter wave} (mm-wave) frequencies, to alleviate the spectrum crunch \cite{rappaport}.
% This increased interest can be attributed to the availability of larger bandwidth in the mm-wave frequency band to support the ever increasing demand for mobile traffic.
% However, channel propagation at these frequencies present significant challenges, such as path-loss and sensitivity to blockages \cite{Singh,Singh2,Bai15}.
% In particular, blockages result in irregular interference patterns over the network, caused by the presence of buildings, foliage, cars, etc.
% \emph{Thus, attenuation is not solely defined by the distance between transmitter and receiver, but is highly affected by the presence of these blockages.
% Effective interference management must take into account these irregularities.}

% \nm{In dense BS deployments, it has been shown that the network is interference-limited \cite{Rebato}.}

To meet this challenge, we develop and analyze spectrum utilization and interference management techniques for \emph{dense} cognitive radios with \emph{irregular} interference patterns.
 We consider a multi-cell network with a set of PUs and a dense set of opportunistic SUs, which seek access to locally unoccupied spectrum. The SUs must estimate the channel occupancy of the PUs across the network based on local measurements. In principle, these measurements can be collected at a fusion center \cite{Letaief,Ding,Ejaz}, but centralized estimation may incur unacceptable delays and overhead \cite{Goeckel,Wu2012}.
To reduce this cost and provide a form of coordination, neighboring cells may inform each other of spectrum they are occupying \cite{Vasilakos}; however, this scheme cannot
 manage interference beyond the cell neighborhood,
which may be significant in dense topologies.

 \begin{figure}[t]
\centering  
\includegraphics[width=.5\linewidth,trim = 0mm 0mm 0mm 0mm,clip=true]{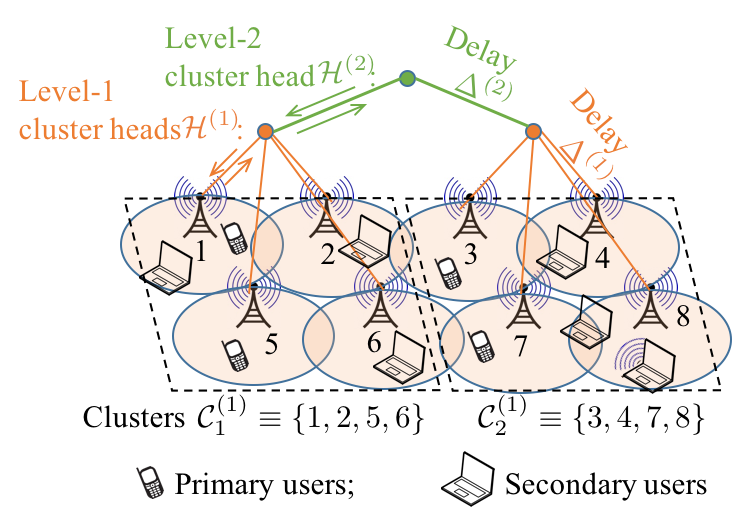}
\caption{System model (see notation in  Sec. \ref{sec:est}).}
\label{fig:sysmo}
\end{figure}

% The network is arranged into cells. In each cell, PUs join and leave the channel at random times. %, thus the state of each cell is described by a first-order binary Markov process.
%  In order to utilize the unoccupied spectrum, the SUs require accurate estimates of spectrum occupancies throughout the cellular network. In principle, 
%  the channel occupancies can be sensed locally in each cell and collected at a fusion center; the global network state information (NSI) collected at the fusion center is then broadcast
%  to each cell for
%  local decision making. In practice, however, such centralized estimation is extremely costly in terms of transmit energy and time.
% Additionally, delays in the acquisition of NSI may offset its benefit, by making state estimates outdated and uninformative for network control.

%In order to reduce the cost incurred in the acquisition of network state information, in this paper we devise a hierarchical approach to spectrum estimation tailored to the spatial nature of wireless interference.
We address this challenge by designing a cost-effective {\em multi-scale} solution to detect and leverage spatio-temporal spectrum access opportunities across the network, by exploiting the structure and irregularities of interference. 
To do so, note that the interference caused by a given SU depends on its position in the network, as depicted in Fig. \ref{fig:sysmo}: PUs closer to this SU will experience stronger interference than PUs farther away. Therefore,
 such SU should estimate more accurately the state of nearby PUs,
 in order to perform more informed local control decisions to access the spectrum or remain idle. In contrast,
 the state of PUs farther away, which experience less interference from such SU, is less relevant to these control decisions, hence coarser spectrum estimates may suffice.
  With this in mind, the goal of our formulation is the design of a cost-effective spectrum sensing architecture to aid local network control,
   which enables each SU to estimate the spectrum occupancy at different spatial scales (hence the name "multi-scale"), so as to 
 possess an accurate and fine-grained estimate of the occupancy of PUs in the vicinity, and coarser estimates of the occupancy states of PUs farther away.
 To achieve this goal, we use a \emph{hierarchical} estimation approach
 resilient to delays and errors in the information exchange and estimation processes,
  inspired by \cite{nokleby:JSTSP13} in the context of averaging consensus \cite{benezit:IT10}:
local measurements are fused hierarchically up a tree, which provides aggregate spectrum occupancy information for clusters of cells at larger and larger scales. 
Thus, SUs acquire precise information on the spectrum occupancies of nearby cells --
these cells are more susceptible to interference caused by nearby SUs --
 and coarse, aggregate information on the occupancies of faraway cells.
 By generating spectrum occupancy estimates at multiple spatial scales (i.e., multi-scale), this scheme permits an efficient trade-off of estimation quality, cost of aggregation, estimation delay, and provides a cost-effective means to acquire information most relevant to network control. We derive the ideal estimator of the global spectrum occupancy from the multi-scale measurements, and we design the SU traffic in each cell in a decentralized fashion so as to maximize a trade-off among SU cell throughput, interference caused to PUs, and mutual SU interference.

%Due to path loss, shadowing, and blockage, SUs accessing the channel in one cell cause significant interference to nearby PUs, but negligible interference to distant PUs. Therefore, \emph{each SU needs precise information about the occupancies of nearby cells, but only coarse information about the occupancies of faraway cells.}
%Based on this intuition, we construct a cellular hierarchy which is used to \emph{aggregate} channel measurements over the network at multiple scales.
%Thus, SUs operating in a given cell have precise knowledge about the local state, \emph{aggregate} knowledge of the states of the cells nearby,
%\emph{aggregate and coarser} knowledge of the states of the cells farther away, and so on at multiple scales. %, reflecting the distance dependent nature of wireless interference.
%In our design, we explicitly account for the delay, local estimation errors and cost incurred in the aggregation of NSI over the tree.

%The proposed aggregation scheme accounts for the irregular interference pattern arising from shadowing and blockage, which are especially severe at mm-wave frequencies \cite{Singh,Singh2,Bai15}.
To tailor the aggregation tree to the interference pattern of the network, we design an {\em agglomerative clustering} algorithm \cite[Ch. 14]{friedman:01}. We measure the end-to-end performance in terms of the trade-off among SU cell throughput, interference to PUs, and the cost efficiency of aggregation.
We show numerically that our design achieves a small degradation in SU cell throughput (up to 15\% under a reference interference-to-noise ratio of 0dB experienced at PUs) compared to a scheme with full network state information,
while incurring only one-third of the cost in the aggregation of spectrum estimates across the network.
We show that the proposed interference-matched tree design
based on agglomerative clustering significantly outperforms a random tree design, thus demonstrating that it provides more relevant information for network control.
Finally, we compare our proposed design with the state-of-the-art consensus-based algorithm \cite{li:TVT10},
originally designed for single-cell systems without temporal dynamics in the PU spectrum occupancy,
and we demonstrate the superiority of our scheme thanks to its ability to
leverage the spatial and temporal dynamics of interference in the network
and to provide more meaningful information for network control.

% We design
% an algorithm based on 
% \emph{agglomerative clustering} \cite[Ch. 14]{friedman:01} to determine an
% energy-efficient aggregation tree that matches the interference pattern of the network. We measure performance in terms of the trade-off between SU network throughput, interference to PUs, and 
% energy efficiency of aggregation. 
% Unlike the regular tree construction proposed in \cite{MichelusiGCOM}, our proposed solution can be applied  to irregular topologies and provides fine-tuning of the aggregation cost.
% We demonstrate numerically that
% this tailored hierarchy outperforms the regular tree construction in \cite{MichelusiGCOM} by 10\% in the SU network throughput, for a reference value of the interference to PUs, at up to 1/4th 
% of the energy cost of
% acquiring full network state information from the local cell neighborhood.

%Our methods also apply to sub-6GHz wireless networks and are robust to issues of directionality of interferers and primary receivers.

%This paper provides important extensions over \cite{MichelusiGCOM}, wherein we assumed a regular tree for hierarchical spectrum sensing by assuming that interference is regular and isotropic, and over \cite{MicheICC}, which considers the special case of no aggregation delay over the tree and noiseless estimates. Herein, we address simultaneously tree design based on the irregular structure of interference (vs \cite{MichelusiGCOM}), as well as the more general case with aggregation delay and noisy estimation (vs \cite{MicheICC}).

\emph{Related work:} Consensus-based schemes for spectrum estimation have been proposed in~\cite{li:TVT10,zeng:JSTSP11,Wu2012,Hajihoseini}:
\cite{Wu2012} proposes a mechanism to 
select only the SUs with the best detection performance to reduce the overhead of spectrum sensing; while \cite{Hajihoseini} focuses on the design of \emph{diffusion} methods.
Cooperative schemes with data fusion have been proposed in  \cite{Letaief,Ding,Ejaz}:
 \cite{Letaief} investigates the optimal voting rule and optimal detection threshold;
 \cite{Ding} proposes a robust scheme to filter out abnormal measurements, such as malicious or unreliable sensors;
 \cite{Ejaz} analyses 
and compares hard and soft combining schemes in heterogeneous networks.
However, all these works focus on a scenario with a single PU pair (one cell) and no temporal dynamics in the PU spectrum occupancy state.  
 Instead, we investigate spectrum sensing in \emph{multi-cell networks} with multiple PU pairs and with \emph{temporal dynamics} in the PU occupancy state, giving rise to both spatial and temporal spectrum access opportunities. Similar opportunities have been explored in \cite{Guoru}, but in the context of a single PU, and without consideration of SU scheduling decisions. In contrast, in our paper we investigate the impact of spectrum sensing on scheduling decisions of SUs.
 
 Another important difference with respect to \cite{li:TVT10,zeng:JSTSP11,Wu2012,Hajihoseini,Letaief} (with the exception of \cite{Guoru}) is that we model temporal dynamics in the occupancy states of each PU, as a result of PUs joining and leaving the network at random times; in time-varying settings, the performance of spectrum estimation may be severely affected by delays in the propagation of estimates across the network, so that spectrum estimates may become outdated. We develop a hierarchical estimation approach that compensates for these propagation delays. A setting with temporal dynamics has been proposed in \cite{myTCNC,Guoru}
for a single-cell system, but without consideration of delays.

Finally, \cite{Bazerque} capitalizes on sparsity due to the narrow-band frequency use, and to sparsely located active radios, and develops estimators to enable
 identification of the (un)used frequency bands at arbitrary locations; differently from this work, we develop techniques to track the activity of PUs, and use this information to schedule transmissions of SUs, hence we investigate the interplay between estimation and scheduling tasks, and the role of network state information.
%Other works include , which investigates the low-cost organization of femtocell overlay networks, and \cite{Niu:JSAC15}, which considers scheduling device-to-device communications in small, millimeter-wave cells.\nm{expand these two}

We summarize the contributions of this paper as follows:
\begin{enumerate}
\item We propose a hierarchical framework to aggregate network state information (NSI) over a multi-cell wireless network,
with a generic interference pattern among cells, which enables spectrum estimation at multiple spatial scales, most informative to network control. We study its performance in terms of the trade-off between the SU cell throughput and the interference
caused to the PUs. We design the optimal SU traffic in each cell in a decentralized fashion, so as to manage the interference caused to other PUs and SUs.
\item  We show that the belief of the spectrum occupancy vector is statistically independent across subsets of cells at different spatial scales, and uniform within each subset (Theorem \ref{thm1}), up to a correction factor that accounts for mismatches in the aggregation delays. This result greatly facilitates the estimation of the interference caused to PUs
 (Lemma \ref{lem:exprew}).
\item We address the design of the hierarchical aggregation tree under a constraint on the aggregation cost based on agglomerative clustering \cite[Ch. 14]{friedman:01}  (Algorithm \ref{alg:clustering}).
\end{enumerate}
\emph{Our analysis demonstrates that multi-scale spectrum estimation
using hierarchical aggregation matched to the structure of interference is a much more cost-effective solution than fine-grained network state estimation, and provides more valuable information for network control.
Additionally, it demonstrates the importance of leveraging
 the spatial and temporal dynamics of interference arising in dense multi-cell systems,
 made possible by our multi-scale strategy; in contrast,
consensus-based strategies, which average out the spectrum estimate over multiple cells and over time, 
are unable to achieve this goal and perform poorly in dense multi-cell systems.}

This paper is organized as follows. 
In Sec. \ref{sysmo}, we present the system model.
In Sec. \ref{sec:est}, we present the proposed local and multi-scale estimation algorithms, whose performance is 
analyzed in Sec. \ref{analysis}.
In Sec. \ref{treedesign}, we address the tree design.
In Sec. \ref{numres}, we present numerical results and, in Sec. \ref{conclu}, we conclude this paper.
The main proofs are provided in the Appendix.
Table I provides the main parameters and metrics.

\begin{table*}\begin{center}
\small
\label{tabnot}
{
\begin{tabular}{| c | l | c | l |}
%  \hline\bf  & \bf Meaning &\bf  & \bf Meaning \\\hline
\hline
$\mathcal C$ & set of cells, with $|\mathcal C|=N_C$
&
$b_{i,t}$& occupancy state of cell $i$ at time $t$, ${\in}\{0,1\}$
\\
$\pi_B$ & steady-state distribution $\mathbb P(b_{i,t}=1)$
&
$\mu$ & memory of the Markov chain $\{b_{i,t},t\geq 0\}$
\\
$\phi_{i,j}$ & INR
 generated by tx in cell $i$ to rx in $j$, cf. \eqref{pathloss}
&
%\hline
$a_{i,t}$ & SU traffic in cell $i$ at time $t$, ${\in}[0,M_{i,t}]$
\\
$M_{i,t}$ & \# of SUs in cell $i$ at time $t$
&
%\hline
$\mathcal B_{N}(p)$ & Binomial with $N$ trials and probability $p$
\\
\hline\hline
%$\mathcal C_p\subseteq \mathcal C$ & $p$th cell partition, $p=1,2,\dots, T$
%\\
%\hline
$\mathcal H^{(L)}$ &
level-$L$ cluster heads
 &
%\hline
 $\mathcal H_{m}^{(L)}$%{\subseteq}\mathcal H^{(L)}
 & level-$L$ cluster heads associated to  $m{\in}\mathcal H^{(L+1)}$
\\  
%\hline
$\mathcal C_{k}^{(L)}$ &
cells associated to $k\in\mathcal H^{(L)}$
&
%\hline
%$H_{L}(i)\in\mathcal H_{L}$ & level-$L$ parent of cell $i$\nm{Remove?}\\
%\hline
$\Lambda_{i,j}$ &
h-distance between cells $i$ and $j$, cf. Def. \ref{hdist}
\\
%\hline
&&
$\mathcal D_i^{(L)}$&
 cells at h-distance $L$  from cell $i$, cf. Def. \ref{def2}
\\
\hline\hline
$\delta_{i}^{(L)}$ & delay from cell $i$ to level-$L$ cluster head
&
%\hline
$\Delta_{m}^{(L)}$ & delay between 
 $m{\in}\mathcal H_{n}^{(L-1)}$ 
 \\
&&& and its upper level-$L$ cluster head $n$
 \\
\hline\hline
$\hat b_{i,t}$ & local estimate at cell $i$
 &
$\sigma_{i,t}^{(L)}$&
delay mismatched aggregate estimate at
 \\
& 
&
&
h-distance $L$ from cell $i$, cf. \eqref{sigmadef}
\\
\hline\hline
$\hat r_{i,t}$ & SU cell $i$ throughput lower bound, cf. \eqref{rtt}
&
$I_{S,i}(t)$ & estimated SU interference at cell $i$, cf. \eqref{Ts}
\\
$\iota_{P,i}$ & INR caused by SUs in cell $i$, cf. \eqref{locrew2}-\eqref{locrew3}
&$I_{P,i}$ & estimated PU interference at cell $i$, cf. \eqref{XXX}
\\
$u_{i,t}$& utility function, cf. \eqref{utility}
&$\pi_{i,t}$ & local belief in cell $i$
\\\hline
\end{tabular}
\caption{Table of Notation}}
\end{center}
\end{table*}
 
\section{System Model}
 \label{sysmo}
 \underline{\bf Network Model:}
We consider the network depicted in Fig.~\ref{fig:sysmo}, composed of a multi-cell network of PUs with $N_C$ cells operating in downlink,
indexed by  $\mathcal C{\equiv}\{1,2,\dots, N_C\}$,
 and an unlicensed network of SUs.
 The receivers are located in the same cell as their transmitters,
 so that they receive from the closest access point. 
Transmissions are slotted and occur over frames. %(typically, an integer number of slots).
Let $t$ be the frame index, and $b_{i,t}{\in} \{0,1\}$ be the PU spectrum occupancy of cell $i{\in}\mathcal C$ during frame $t$,
with $b_{i,t}{=}1$ if occupied and $b_{i,t}{=}0$ otherwise. %This is assumed to be constant over the duration of one frame. 
We suppose that $\{b_{i,t},t\geq 0,i\in\mathcal C\}$  are independent and identically distributed (i.i.d.) across cells and evolve according to a two-state Markov chain, as a result of PUs joining and leaving the network at random times.
We define the transition probabilities as
\\\centerline{$
\nu_1\triangleq\mathbb P(b_{i,t+1}{=}1|b_{i,t}{=}0),
\
\nu_0\triangleq\mathbb P(b_{i,t+1}{=}0|b_{i,t}{=}1),
$}\\
where $\mu{\triangleq}1{-}\nu_1{-}\nu_0$
is the \emph{memory} of the Markov chain, 
 which dictates the rate of convergence to its steady-state distribution. 
 Hence, $\pi_B{\triangleq}\mathbb P(b_{i,t}=1){=}\frac{\nu_1}{1{-}\mu}$ 
at steady-state. 
We denote the state of the network at time $t$ as $\mathbf b_t=(b_{1,t},b_{2,t},\dots,b_{N_C,t})$.
\\\indent We assume that PUs and SUs coexist in the same spectrum band. Let $M_{i,t}$ be
 the number of SUs in cell $i$ at time $t$, which may vary over time as a result of
SUs joining and leaving the network. We collect $M_{i,t}$ in the vector $\mathbf M_t$.
\begin{assumption}
 \label{assum1}
$\{\mathbf M_{t}{,}t{\geq}0\}$ are i.i.d. across cells, stationary and independent of $\{\mathbf b_{t}{,}t{\geq}0\}$, that is
\begin{align*}
 &\mathbb P\bigl(
\mathbf b_{t}{=}\tilde{\mathbf b}_{t},
\mathbf M_{t}{=}\tilde {\mathbf M}_{t},\forall t\in\mathcal T
 \bigr)
% \\&
 =
 \prod_{i}
 \mathbb P\bigl(
b_{i,t}{=}\tilde b_{i,t}, t{\in}\mathcal T
 \bigr)
\\&\qquad\times
  \mathbb P\bigl(
M_{i,t}{=}\tilde M_{i,t}, t{\in}\mathcal T
 \bigr)\ \text{ (independence),}
 \\
&
  \mathbb P\bigl(
M_{i,t}{=}\tilde M_{i,t}, t{\in}\mathcal T
 \bigr)
 {=}\mathbb P\bigl(
M_{i,t-\delta}{=}\tilde M_{i,t}, t{\in}\mathcal T
 \bigr)\ \text{ (stationarity).}
\end{align*}
 where $\mathcal T$ is a time interval
 and $\delta>0$ is a delay.
 Additionally, $M_{i,t}>0, \forall i,t$ (dense network).\qed
\end{assumption}
 Assumption \ref{assum1} guarantees that spectrum estimates are ``statistically symmetric''  \cite{MicheTSP1}, \emph{i.e.}, they exhibit the same statistical
 properties at different cells and delay scales. An example which obeys  Assumption \ref{assum1} is when $M_{i,t}$ is a Markov chain taking values from $M_{i,t}>0$, i.i.d. across cells.
The SUs opportunistically access the spectrum to maximize their own cell throughput,
while at the same time limiting the interference caused to other SUs and to the PUs.
Their access decision is governed by the \emph{local SU access traffic} $a_{i,t}{\in}[0,M_{i,t}]$ for SUs in cell $i$.
We assume an uncoordinated SU access strategy so that,
given $a_{i,t}$, all the $M_{i,t}$ SUs in cell $i$ access the channel with probability $a_{i,t}/M_{i,t}$, independently of each other.\footnote{We assume that $M_{i,t}$ is known in cell $i$, and a local control channel is available to regulate the local SU traffic $a_{i,t}$.}
Therefore, $a_{i,t}$ represents the expected number of SU transmissions in cell $i$.
 %where $a_{i,t}{=}1$ if the SUs operating in cell $i$ access the channel at time $t$, and  $a_{i,t}{=}0$ otherwise.
%{If $a_{i,t}{=}1$, each SU in cell $i$ transmits with probability $p_{tx}$ at time $t$.}
 We let $\mathbf a_t{=}(a_{1,t},a_{2,t},\dots,a_{N_C,t})$.
%The activity of the SUs generate interference to the cellular network.
\\
\indent Transmissions of SUs and PUs generate interference to each other.
We denote the interference to noise ratio (INR) %temperature \cite{Kolodzy}
 generated by the activity of 
a transmitter in cell $i$ to a receiver in $j$ as $\phi_{i,j}{\geq}0$,
%We assume that interference is symmetric, so that $\phi_{i,j}=\phi_{j,i},\forall i,j\in\mathcal C$. While this assumption is commonly employed and is based on channel reciprocity, the analysis extends in a straightforward  manner to asymmetric interference. \nm{Do we really need this?}
%Note that $\phi_{i,i}$ is the  strength of the interference caused by the SUs in cell $i$ to cell $i$.
 collected into the symmetric (due to channel reciprocity) matrix  $\boldsymbol{\Phi}{\in} \mathbb{R}^{N_C \times N_C}$.
 Typically, 
 \begin{align}
 \label{pathloss}
& [\phi_{i,j}]_{\mathrm{dB}}=[P_{tx}]_{\mathrm{dBm}}-[N_0W_{tot}]_{\mathrm{dBm}}
  \nonumber\\&\qquad
 -[L_{ref}]_{\mathrm{dB}}-\alpha_{i,j}[d_{i,j}/d_{ref}]_{\mathrm{dB}}
 \end{align}
 (see, e.g., \cite{Sun}), where
 $P_{tx}$ is the transmission power, common to all PUs and SUs,
 $N_0$ is the noise power spectral density and $W_{tot}$ is the signal bandwidth;
 $L_{ref}$ is the large-scale pathloss at a reference distance $d_{ref}$, based on Friis' free space pathloss formula, and $[d_{i,j}/d_{ref}]^{\alpha_{i,j}}$ is the distance dependent component, with $d_{i,j}$ and  $\alpha_{i,j}$ the distance and pathloss exponent between cells $i$ and $j$.
 We assume that the intended receiver of each PU or SU transmission is located within the cell radius, so that $\phi_{i,i}$ is the SNR  to the intended receiver in cell $i$.
 In practice, the large-scale pathloss exhibits variations as transmitter or receiver are moved within the cell coverage. Thus, $\phi_{i,j}$ can be interpreted as an average of
  these pathloss variations, or a low resolution approximation of the large-scale pathloss map.
 This is a good approximation due to the small cell sizes arising in dense cell deployments,  as considered in this paper. In Sec. \ref{numres} (Fig. \ref{fig:simresreal}), we will 
 demonstrate its robustness in a more realistic setting.
\\\indent\underline{\bf Network Performance Metrics:}
{We label each SU as $(j,n)$, denoting the $n$th SU in cell $j$.}
{Let $v_{j,n,t}\in\{0,1\}$ be the indicator of whether SU $(j,n)$ transmits
based on the probabilistic access decision outlined above; this is stacked in the vector $\mathbf v_t$.
If the reference SU $(i,1)$ transmits,} the signal received by the corresponding SU receiver is
\begin{align}
&y_{i,1}(t)=\sqrt{\phi_{i,i}}h_{i,1}^{(s)}(t) x_{i,1}^{(s)}(t)
+{w_{i,1}(t)}+ n_{i,1}(t),
\label{signalmodel}
\end{align}
where we have defined the interference signal
\begin{align}
w_{i,1}(t)\triangleq&
\sum_{(j,n)\neq (i,1)}\sqrt{\phi_{j,i}}h_{j,n}^{(s)}(t)v_{j,n,t} x_{j,n}^{(s)}(t)
  \nonumber\\&
+\sum_{j=1}^{N_C}\sqrt{\phi_{j,i}}h_{j}^{(p)}(t)b_{j,t} x_{j}^{(p)}(t),
\end{align}
$h_{j,n}^{(s)}(t)$ is the fading channel between
SU $(j,n)$ and the reference SU $(i,1)$, with $ x_{j,n}^{(s)}(t)$ the unit energy transmitted signal;
$h_{j}^{(p)}(t)$ is the fading channel between
PU $j$ (transmitting in downlink) and SU $(i,1)$, with $ x_{j}^{(p)}(t)$ the  unit energy transmitted signal; 
$\phi_{j,i}$ is the large-scale pathloss between cells $j$ and $i$, see \eqref{pathloss};
$n_{i,1}(t){\sim}\mathcal{CN}(0,1)$ is circular Gaussian noise;
 we assume Rayleigh fading, so that 
 $h_{j,n}^{(s)}(t),h_{j}^{(p)}(t)\sim\mathcal{CN}(0,1)$. The transmission is successful if and only if the SINR exceeds a threshold $\mathrm{SINR}_{\mathrm{th}}$;
 we then obtain the success probability of SU $(i,1)$,
 conditional on $\mathbf v_t$ and $\mathbf b_t$,
 \begin{align}
&\rho_{i,1}(\mathbf v_t,\mathbf b_t)=
\mathbb P\left(\left.
\frac{\phi_{i,i}|h_{i,1}^{(s)}(t)|^2}{1+\left|w_{i,1}(t)\right|^2}
>\mathrm{SINR}_{\mathrm{th}}
\right|\mathbf v_t,\mathbf b_t\right).
\end{align}
Noting that
$w_{i,1}(t)|(\mathbf v_t,\mathbf b_t)$ is circular Gaussian with zero mean and variance
\begin{align}
\mathbb E[|w_{i,1}(t)|^2|\mathbf v_t,\mathbf b_t]\triangleq
\sum_{j=1}^{N_C}\phi_{j,i}\eta_j+\sum_{j=1}^{N_C}\phi_{j,i}b_{j,t}-\phi_{i,i},
\end{align}
where $\eta_j\triangleq\sum_{n=1}^{M_{j,t}}v_{j,n,t}$ is the number of SUs that attempt
spectrum access in cell $j$,
we obtain
 \begin{align*}
&\rho_i(\mathbf v_t,\mathbf b_t){=}
\frac{
e^{-\mathrm{SINR}_{\mathrm{th}}/\phi_{i,i}}
}{1+\mathrm{SINR}_{\mathrm{th}}
\left[
\sum_{j=1}^{N_C}\frac{\phi_{j,i}}{\phi_{i,i}}\eta_j
+\sum_{j=1}^{N_C}\frac{\phi_{j,i}}{\phi_{i,i}}b_{j,t}-1
\right].}
\end{align*}
Then, the throughput in cell $i$, conditional on the SU traffic $\mathbf a_t$ and PU network state $\mathbf b_t$,
 is obtained by %averaging over the access decisions $\mathbf v_t$.
 noting that  each of the $\eta_i$ SUs succeed with probability $\rho_i$;
 hence, taking the expectation with respect to the number of
SUs performing spectrum access, $\eta_j{\sim}\mathcal B_{M_{j,t}}(a_{j,t}/M_{j,t})$
(binomial random variable with probability $a_{j,t}/M_{j,t}$ and $M_{j,t}$ trials), we obtain
\begin{align}
\label{riatbt}
&r_{i,t}(\mathbf a_t,\mathbf b_t)
\\&\nonumber
{\triangleq}
\mathbb E_{\eta}\left[
\left.
\frac{
\eta_i\exp\left\{-\frac{1}{\phi_{i,i}}\mathrm{SINR}_{\mathrm{th}}\right\}
}{1{+}\mathrm{SINR}_{\mathrm{th}}
\left[
\sum_{j=1}^{N_C}\frac{\phi_{j,i}}{\phi_{i,i}}\eta_j{-}1
{+}\sum_{j=1}^{N_C}\frac{\phi_{j,i}}{\phi_{i,i}}b_{j,t}
\right]
}
\right|\mathbf a_t,\mathbf b_t
\right]
\\&
\nonumber
{=}
\mathbb E_{\eta,\hat\eta_i}\left[
\left.
\frac{a_{i,t}\exp\left\{-\frac{1}{\phi_{i,i}}\mathrm{SINR}_{\mathrm{th}}\right\}
}{1{+}\mathrm{SINR}_{\mathrm{th}}
\left[
\hat\eta_i{+}\sum_{j\neq i}\frac{\phi_{j,i}}{\phi_{i,i}}\eta_j
{+}\sum_{j=1}^{N_C}\frac{\phi_{j,i}}{\phi_{i,i}}b_{j,t}
\right]
}
\right|\mathbf a_t,\mathbf b_t
\right],
\end{align}
where the second equality is obtained by the change of variable $\hat\eta_i{=}\eta_i{-}1$, with $\hat\eta_i{\sim}\mathcal B(a_{i,t}/M_{i,t},M_{i,t}{-}1)$.
 The computation of the SU cell throughput using this formula has high complexity, due to the outer expectation.
 Therefore, we resort to a lower bound.
Noting that the argument of the expectation is a convex function of $\eta_j,\forall j$ and $\hat\eta_i$, Jensen's inequality yields
\begin{align*}
r_{i,t}(\mathbf a_t,\mathbf b_t)
{\geq}&\frac{a_{i,t}\exp\left\{-\frac{1}{\phi_{i,i}}\mathrm{SINR}_{\mathrm{th}}\right\}
}{1{+}\mathrm{SINR}_{\mathrm{th}}\sum_{j=1}^{N_C}\frac{\phi_{j,i}}{\phi_{i,i}}(a_{j,t}+b_{j,t})
{-}\mathrm{SINR}_{\mathrm{th}}\frac{a_{i,t}}{M_{i,t}}
}.
\end{align*}
 Cell $i$ selects $a_{i,t}$ based on partial NSI, denoted by the local belief $\pi_{i,t}(\mathbf b)$ 
that $\mathbf b_t{=}\mathbf b$.
Taking the expectation over $\mathbf b_t$
conditional on $\pi_{i,t}$ and using Jensen's inequality, we obtain
\begin{align}
\label{rtt}
&\mathbb E\left[r_{i,t}(\mathbf a_t,\mathbf b_t)|\pi_{i,t}\right]
\nonumber\\&
\geq\frac{a_{i,t}\exp\left\{-\frac{1}{\phi_{i,i}}\mathrm{SINR}_{\mathrm{th}}\right\}
}{1+\mathrm{SINR}_{\mathrm{th}}
\left[a_{i,t}(1-M_{i,t}^{-1})+I_{P,i}(\pi_{i,t})+I_{S,i}(t)\right]
}
\nonumber\\&
\triangleq\hat r_{i,t}(a_{i,t},I_{P,i}(\pi_{i,t})),
\end{align} 
where we have defined
\begin{align}
\label{Ts}
&I_{S,i}(t)\triangleq\sum_{j\neq i}\frac{\phi_{j,i}}{\phi_{i,i}}a_{j,t},
\\&
\label{XXX}
I_{P,i}(\pi_{i,t})
{\triangleq }
\mathbb E\left[
\left.\sum_{j=1}^{N_C}\frac{\phi_{j,i}}{\phi_{i,i}}b_{j,t}\right|\pi_{i,t}
\right]
{=}
\sum_{j=1}^{N_C}\frac{\phi_{j,i}}{\phi_{i,i}}\mathbb P\left(b_{j,t}{=}1|\pi_{i,t}\right).
\end{align}
The terms $I_{S,i}(t)$ and $I_{P,i}(\pi_{i,t})$ represent, respectively, an estimate of the interference strength caused by SUs and PUs operating in the rest of the network to the reference SU in cell $i$.
Additionally, due to channel reciprocity and the resulting symmetry on $\boldsymbol{\Phi}$, $I_{P,i}(\pi_{i,t})$ represents an estimate of the interference strength caused by the reference SU to the rest of the PU network.

Herein, we use $\hat r_{i,t}$ in \eqref{rtt} to characterize the performance of the SUs. Since this is a lower bound to the actual SU cell throughput,
using $\hat r_{i,t}$ as a metric provides performance guarantees. Note that the performance depends upon the network-wide SU activity $\mathbf a_t$ via $I_{S,i}(t)$; in turn, each $a_{j,t}$
is decided based on the local belief $\pi_{j,t}$, which may be unknown to the SUs in cell $i$ (which operate under a different belief $\pi_{i,t}$).
Therefore, maximization of $\hat r_{i,t}(a_{i,t},I_{P,i}(\pi_{i,t}))$ can be characterized as a \emph{decentralized decision} problem, which 
does not admit polynomial time algorithms \cite{Bernstein}. To achieve low computational complexity, we relax the decentralized decision process by assuming that 
 $I_{S,i}(t)$ is known to cell $i$ in slot $t$. This assumption is based on the
 following practical arguments:
due to the Markov chain dynamics of $\mathbf b_t$, $\mathbf a_t$ varies slowly over time, hence  $I_{S,i}(t)$ can be estimated by averaging the SU traffic over time;
additionally, the spatial variations of $\mathbf a_t$ are averaged out in the spatial domain since 
$I_{S,i}(t)$ is a weighted sum of $a_{j,t}$ across cells, yielding slow variations on 
$I_{S,i}(t)$ due to mean-field effects.
In Sec. \ref{analysis}, we will present an approach to estimate $I_{S,i}(t)$ and $I_{P,i}(\pi_{i,t})$ based on hierarchical information exchange over the SU network.

We define
the average INR experienced by the PUs as a result of the activity of the SUs as
\begin{align}
\label{INR}
\mathrm{INR}(\mathbf a_t,\mathbf b_t)\triangleq\frac{1}{N_C\pi_B}\sum_{j=1}^{N_C}\sum_{i=1}^{N_C}a_{i,t}\phi_{i,j}b_{j,t},
\end{align}
where $N_C\pi_B$ is the average number of  active PUs at steady-state.
In fact, the expected number of SUs transmitting in cell $i$ is $a_{i,t}$, so that 
$a_{i,t}\phi_{i,j}$ is 
the overall interference caused by SUs in cell $i$ to the PU in cell $j$.
$\mathrm{INR}(\mathbf a_t,\mathbf b_t)$ is then obtained by averaging this effect over the  network.
Herein, we isolate the contribution due to the SUs in cell $i$ on
\eqref{INR}, yielding
\begin{align}
&\iota_{P,i}(a_{i,t},\mathbf b_t)\triangleq \frac{1}{\pi_B}a_{i,t}\sum_{j=1}^{N_C} \phi_{i,j}b_{j,t},
\label{locrew2}
\end{align}
so that $\mathrm{INR}(\mathbf a_t,\mathbf b_t)\triangleq\frac{1}{N_C}\sum_{i=1}^{N_C}
\iota_{P,i}(a_{i,t},\mathbf b_t)$.
By computing the expectation with respect to the local belief $\pi_{i,t}$ and using the symmetry of $\boldsymbol{\Phi}$, we then obtain
\begin{align}
\iota_{P,i}(a_{i,t},I_{P,i}(\pi_{i,t}))&\triangleq
\mathbb E[\iota_{P,i}(a_{i,t},\mathbf b_t)|\pi_{i,t}]
\nonumber\\&
=
 \frac{1}{\pi_B}a_{i,t}\phi_{i,i}I_{P,i}(\pi_{i,t}).
\label{locrew3}
\end{align}
\indent Since the goal of SUs is to maximize their own cell throughput, while minimizing their interference to the PUs, 
we define the local utility as a \emph{payoff minus cost} function,
\begin{align}
\label{utility}
&u_{i,t}(a_{i,t},I_{P,i}(\pi_{i,t}))
\nonumber\\&
\triangleq\hat r_{i,t}(a_{i,t},I_{P,i}(\pi_{i,t}))-\lambda\iota_{P,i}(a_{i,t},I_{P,i}(\pi_{i,t})),
\end{align} 
where $\lambda{>}0$ is a cost parameter which balances the two competing goals.
Given $\pi_{i,t}$, the goal of the SUs in cell $i$ is to design $a_{i,t}$ so as to maximize
$u_{i,t}(a_{i,t},I_{P,i}(\pi_{i,t}))$. Since this is a concave function of $a_{i,t}$ (as can be seen by inspection), we obtain
the optimal SU traffic
\begin{align}
\label{rewardinb}
\nonumber
&a_{i,t}^*(I_{P,i}(\pi_{i,t}))\triangleq\underset{a_{i,t}\in[0,M_{i,t}]}{\arg\max}u_{i,t}(a_{i,t},I_{P,i}(\pi_{i,t}))
\nonumber\\&
=
\Biggl[\frac{
\sqrt{1+\mathrm{SINR}_{\mathrm{th}}\left[I_{P,i}(\pi_{i,t})+I_{S,i}(t)\right]}
}{\mathrm{SINR}_{\mathrm{th}}(1-1/M_{i,t})}
\\&{\times}
\Biggl(
\frac{\sqrt{\pi_B}e^{-\frac{\mathrm{SINR}_{\mathrm{th}}}{2\phi_{i,i}}}}{\sqrt{\lambda\phi_{i,i}I_{P,i}(\pi_{i,t})}}
{-}\sqrt{1{+}\mathrm{SINR}_{\mathrm{th}}\left[I_{P,i}(\pi_{i,t}){+}I_{S,i}(t)\right]}
\Biggr)\Biggr]_{0}^{M_{i,t}},
\nonumber
\end{align}
where $[\cdot]_0^m=\min\{\max\{\cdot,0\},m\}$ denotes the projection operation 
onto the interval $[0,m]$. It can be shown by inspection that
 both $a_{i,t}^*$ and $u_{i,t}^*$ are non-increasing functions of $I_{P,i}(\pi_{i,t})$, so that, as the PU activity increases ($I_{P,i}(\pi_{i,t})$ increases), the SU activity and the local utility both decrease;
 when $I_{P,i}(\pi_{i,t})$ is above a certain threshold,
 then $a_{i,t}^*=0$ and $u_{i,t}^*(I_{P,i}(\pi_{i,t}))=0$; indeed, in this case the PU network experiences high activity, hence SUs remain idle to
avoid interfering.
Additionally,  $u_{i,t}^*(I_{P,i}(\pi_{i,t}))$ is a convex function of 
$I_{P,i}(\pi_{i,t})$. Then, by Jensen's inequality,
\begin{align}
\label{upbound}
u_{i,t}^*(I_{P,i}(\pi_{i,t}))\leq\sum_{\mathbf b\in\{0,1\}^{N_C}}
\pi_{i,t}(\mathbf b)
u_{i,t}^*(I_{P,i}(\mathcal I_{\mathbf b})),
\end{align}
where $\mathcal I_{\mathbf b}$ is the Kronecker delta function centered at $\mathbf b$,
reflecting the special case when $\mathbf b_t$ is known, so that
$u_{i,t}^*(I_{P,i}(\mathcal I_{\mathbf b}))$ represents the utility achieved when 
$\mathbf b_t=\mathbf b$, known.
Consequently, the expected network utility is maximized when $\mathbf b_t$ is known
 (full NSI).
Thus, the SUs should, possibly, obtain full NSI in order to achieve the best performance.
To approach this goal,
the SUs in cell $i$ should obtain $\mathbf b_t$ in a timely fashion.
To this end, the SUs in cell $j{\neq}i$ should report the local
and current spectrum state $b_{j,t}$ to the SUs in cell $i$ via information exchange, potentially over multiple hops.
Since this needs to be done over the entire network (\emph{i.e.}, for every pair $(i,j)\in\mathcal C^2$), the associated overhead may be impractical in dense multi-cell network deployments. 
Additionally, these spectrum estimates may be noisy and delayed, hence they may become outdated and not informative for network control.
In order to reduce the overhead of full-NSI, we now develop a scheme to estimate spectrum occupancy based on \emph{delayed, noisy, and aggregate} (vs
timely, noise-free and fine-grained)
spectrum measurements  over the network.

\section{Local and Multi-scale Estimation Algorithms}
\label{sec:est}
In this section, we propose a method to estimate $I_{P,i}(\pi_{i,t})$ and $I_{S,i}(t)$
at cell $i$ based on hierarchical information exchange. To this end, SUs exchange estimates of the local PU spectrum occupancy $b_{i,t}$, denoted as $\hat b_{i,t}$,
as well as the local SU traffic decision variable $a_{i,t}$.
For conciseness, we will focus on the estimation of $I_{P,i}(\pi_{i,t})$ in this section; however,
the same technique can be applied straightforwardly to the estimation of $I_{S,i}(t)$ as well.
 In fact, $I_{P,i}(\pi_{i,t})$ and $I_{S,i}(t)$ have the same structure -- they both are 
a weighted sum of the respective local variables $\mathbb E[b_{i,t}|\pi_{i,t}]$ and  $a_{i,t}$, with weights $\frac{\phi_{j,i}}{\phi_{i,i}}$, see \eqref{Ts}-\eqref{XXX},
hence they can be similarly estimated.
\subsection{Aggregation tree}
To reduce the cost of acquisition of NSI, we propose a  
{\em multi-scale} approach to spectrum sensing.
To this end, we partition the cell grid into $P$ sets
$\mathcal C_p,p{=}1,\dots,P$, 
and define a tree on each $\mathcal C_p$, designed in Sec. \ref{treedesign}.
  Since each edge in the tree incurs delay, $P$ disconnected trees are equivalent to a single tree where the 
edges connecting each of the $P$ subtrees to the root have \emph{infinite} delay (and thus, provide outdated, non-informative NSI). 
Hence, without loss of generality, we assume $P{=}1$ where, possibly, some edges incur infinite delay.

Level-$0$ contains the leaves, represented by the cells $\mathcal C$.
To each cell, we associate the singleton set $\mathcal C_{i}^{(0)}{\equiv}\{i\},i{\in}\mathcal C$.\footnote{Note that
$\mathcal C_{i}^{(0)}$ represents cell $i$, containing $M_{i,t}{>}0$ SUs (Assumption~\ref{assum1}).}
At level-$1$, let $\mathcal C_{k}^{(1)},1{\leq}k{\leq}n^{(1)}$ be a partition of $\mathcal C$ into $n^{(1)}{\leq}|\mathcal C|$ non-empty subsets,
each associated to a cluster head $k$.
The set of $n^{(1)}$ level-$1$ cluster heads is denoted as $\mathcal H^{(1)}$.
Hence, $\mathcal C_{k}^{(1)}$  is the set of cells associated to the level-1 cluster head $k{\in}\mathcal H^{(1)}$,
see Fig. \ref{fig:sysmo}.

Recursively, at level-$L$, let $\mathcal H^{(L)}$ be the set of  level-$L$ cluster heads,
 with $L{\geq}1$. If  $|\mathcal H^{(L)}|{=}1$, then 
 we have defined a tree with depth $D{=}L$.
 Otherwise, we define a partition of $\mathcal H^{(L)}$ into
$n^{(L{+}1)}{\leq}|\mathcal H^{(L)}|$ non-empty subsets $\mathcal H_{m}^{(L)},m{=}1,\dots, n^{(L{+}1)}$, each associated to a level-$(L{+}1)$ cluster head, collected in the set
 $\mathcal H^{(L{+}1)}{\equiv}\{1,\dots,n^{(L{+}1)}\}$.
Let $\mathcal C_{m}^{(L{+}1)}$ be the set of cells associated to level-$(L{+}1)$ cluster head $m\in\mathcal H^{(L{+}1)}$.
This is obtained recursively as
\begin{align}
\label{recC}
\mathcal C_{m}^{(L+1)}=\bigcup_{k\in\mathcal H_{m}^{(L)}}\mathcal C_{k}^{(L)},\ \forall m\in\mathcal H^{(L+1)}.
\end{align}

%Let $H_{L}(i)\in\mathcal H_{L}$ be the level-$L$ parent of cell $i\in\mathcal C$,
%\emph{i.e.}, $H_{0}(i)=i$, and $H_{L}(i)=k$ for $L\geq 1$ if and only if $i\in\mathcal C_{L}^{(k)}$, for some $k\in\mathcal H_{L}$.
We are now ready to state some important definitions.
\begin{definition}
\label{hdist}
We define
the \emph{hierarchical distance} (h-distance) between cells $i,j{\in}\mathcal C$ as
\\
\centerline{$
\hfill\Lambda_{i,j}\triangleq\min\left\{L\geq 0:
{i,j\in\mathcal C_{m}^{(L)},\exists m\in\mathcal H^{(L)}}\right\}.\hfill\qed
$}
%\centerline{$\pushQED{\qed} \Lambda_{i,j}\triangleq\min\left\{L\geq 0:{i,j\in\mathcal C_{m}^{(L)},\exists m\in\mathcal H^{(L)}}\right\}.\qedhere\popQED$}
\end{definition}
\noindent In other words, $\Lambda_{i,j}$ is the smallest level of the cluster containing both $i$ and $j$.
% lowest level $L$ such that cells $i$ and $j$ belong to the same level-$L$ cluster. 
It follows that 
the  h-distance between cell $i$ and itself is $\Lambda_{i,i}{=}0$, and it is symmetric ($\Lambda_{i,j}{=}\Lambda_{j,i}$).
\begin{definition}
\label{def2}
Let $\mathcal D_i^{(L)}$ be the set of cells at h-distance $L$  from cell $i$: $\mathcal D_i^{(0)}{\equiv}\{i\}$, and, for all $m\in\mathcal H^{(L)}$, $k\in\mathcal H_m^{(L{-}1)}$,
$i\in\mathcal C_{k}^{(L{-}1)}$ (then, $k$ is the level-$(L{-}1)$ cluster head of cell $i$)
\\
\centerline{\hfill$
\mathcal D_i^{(L)}\equiv
\mathcal C_{m}^{(L)}\setminus\mathcal C_{k}^{(L-1)},\ L>0.
$\hfill\qed}
\end{definition}
In fact, $\mathcal C_{m}^{(L)}$ contains all cells at h-distance (from cell $i$) less than (or equal to) $L$.
Thus, we obtain $\mathcal D_i^{(L)}$ by removing from $\mathcal C_{m}^{(L)}$ all cells  at h-distance less than (or equal to)  $L-1$, $\mathcal C_{k}^{(L-1)}$ (note that this is a subset of
$\mathcal C_{m}^{(L)}$, since $k\in\mathcal H_m^{(L-1)}$).
{For example, with reference to Fig. \ref{fig:sysmo},
$\mathcal D_1^{(0)}\equiv\{1\}$ (cell $1$ is at  h-distance $0$ from itself),
$\mathcal D_1^{(1)}\equiv\{2,5,6\}$ (cells $2$, $5$ and $6$ are at h-distance $1$ from cell $1$),
$\mathcal D_1^{(2)}\equiv\{3,4,7,8\}$ (cells $3$, $4$, $7$ and $8$ are at h-distance $2$ from cell $1$).}

\subsection{Local Estimation}
\label{Localestimation}
The first portion of the frame is used by SUs for spectrum sensing, the remaining 
 portion for data communication.
Thus,  spectrum sensing does not suffer from SU interference.
\begin{remark}
This frame structure requires accurate synchronization among SUs, 
achievable using techniques developed in \cite{Rentel}.
Loss of synchronization may cause overlap between the sensing and communication phases; herein, we assume that
the duration of the sensing phase is sufficiently larger than synchronization errors, so that this overlap is negligible.
\end{remark}
In the spectrum sensing portion of frame $t$, $M_{i,t}$ SUs in cell $i$ estimate
 $b_{i,t}$.
Each of the $M_{i,t}$ SUs
observe the local state $b_{i,t}$ through a 
binary asymmetric channel, $\mathrm{BC}(\epsilon_{F},\epsilon_M)$, 
where $\epsilon_F$ is the false-alarm probability ($b_{i,t}{=}0$ is detected as being occupied)
and $\epsilon_M$ is the mis-detection probability ($b_{i,t}{=}1$ is detected as being unused).
In practice, each SU measures the received energy level and compares it to a threshold; the value of this threshold
entails a trade-off between $\epsilon_F$ and  $\epsilon_M$.
We assume that these $M_{i,t}$ spectrum measurements 
are i.i.d. across SUs (given $b_{i,t}$). %In addition to accounting for measurement errors in energy detection, the error probabilities account for hidden node and cell ``blindness'' issues.  %the latter of which are more pronounced in highly directional propagation environments.
 In principle, $\epsilon_F,\epsilon_M$ may vary over cells and time, but for simplicity we treat them as constant.

Then, these measurements are fused at a local fusion center at the cell level,\footnote{The optimal design of decision threshold, local estimators, fusion rules, are outside the scope of this paper and can be found in other prior work, such as \cite{Letaief,Ding,Ejaz}, for the case of a single-cell.} and then up the hierarchy,
using an out-of-band channel which does not interfere with PUs.
Thus, the number of 
measurements that detect (possibly, with errors) the spectrum as occupied in cell $i$,
denoted as $\xi_{i,t}{\in}\{0,\dots,M_{i,t}\}$, is a sufficient statistic to estimate $b_{i,t}$.
Let
\\
\centerline{$\bar b_{i,t}\triangleq\mathbb P(b_{i,t}=1|\text{past measurements})$}
\\
be the prior probability of occupancy of cell $i$, time $t$, given measurements collected up to $t$ (excluded). After collecting the $M_{i,t}$ measurements, the cell head estimates $b_{i,t}$ as
\begin{align}
\label{post_estimate}
\hat b_{i,t}&\triangleq\mathbb P(b_{i,t}=1|\text{past measurements},\xi_{i,t}=\xi)
\\&
=
\frac{\bar b_{i,t}\mathbb P(\xi_{i,t}=\xi|b_{i,t}=1)}
{
\bar b_{i,t}\mathbb P(\xi_{i,t}=\xi|b_{i,t}=1)
+(1-\bar b_{i,t})\mathbb P(\xi_{i,t}=\xi|b_{i,t}=0)
},
\nonumber
\end{align}
where the second step follows from Bayes' rule.
Note that $[\xi_{i,t}|b_{i,t}{=}1]{\sim}\mathcal B_{M_{j,t}}(1{-}\epsilon_M)$
and $[\xi_{i,t}|b_{i,t}{=}0]{\sim}\mathcal B_{M_{j,t}}(\epsilon_F)$. Thus, we obtain
\begin{align}
\label{hatbit}
\hat b_{i,t}{=}
\frac{\bar b_{i,t}\left(1-\epsilon_M\right)^{\xi}\epsilon_M^{M_{i,t}-\xi}}
{
\bar b_{i,t}\left(1{-}\epsilon_M\right)^{\xi}\epsilon_M^{M_{i,t}{-}\xi}
{+}(1{-}\bar b_{i,t})\epsilon_F^{\xi}\left(1{-}\epsilon_F\right)^{M_{i,t}{-}\xi}
}.
\end{align}
Given $\hat b_{i,t}$,
the prior probability in the next frame is obtained based on the spectrum occupancy dynamics as
\begin{align}
\nonumber
&\bar b_{i,t+1}\triangleq\mathbb P(b_{i,t+1}=1|\text{past measurements},\xi_{i,t}=\xi)
\\&
=(1-\nu_0)\hat b_{i,t}+\nu_1(1-\hat b_{i,t})
=(1-\mu)\pi_B+\mu\hat b_{i,t}.
\end{align}
%where we have used the fact that $\nu_1=(1-\mu)\pi_B$.
% In the special case $M_{i,t}=0$, no SUs estimate the local spectrum occupancy $b_{i,t}$, hence $\xi_{i,t}=0$, yielding
%\begin{align}
%	\hat b_{i,t}=\bar b_{i,t}, \quad \bar b_{i,t+1}=(1-\mu)\pi_B+\mu\bar b_{i,t}.
%\end{align}

%In the rest of this section, without loss of generality, we consider a reference cell $i\in\mathcal C_p$, for some $p=1,2\dots,T$. Thus, for notational simplicity, we neglect any dependence of the variables defined on $i$ and $p$.

\subsection{Hierarchical information exchange over the tree}
In the previous section, we discussed the local estimation at the cell level.
 %In order to collect NSI, the SUs exchange these local estimates over the tree.
 We now describe the {\em hierarchical} fusion of local estimates to collect multi-scale NSI. 
 This fusion is patterned after {\em hierarchical averaging} \cite{nokleby:JSTSP13}, a technique for scalar average consensus in wireless networks.

The aggregation process running at each node is depicted in Fig. \ref{fig:aggalgo}. 
The cell head, after the local spectrum sensing in frame $t$, has a local spectrum estimate $\hat b_{i,t}$.
These local estimates are fused up the hierarchy, incurring delay.
 % due to multi-hop transmissions and local processing, before being propagated back down to the individual cells.
  Let $\delta_i^{(L)}\geq 0$ be the delay to propagate the spectrum estimate of cell $i$ all the way up to its
 level-$L$ cluster head $n$. It includes the local processing time at each intermediate level-$l$ cluster head traversed before reaching the level-$L$ cluster head, as well as the delay to traverse the links (possibly, multi-hop) connecting successive cluster heads.
 We assume that $\delta_i^{(L)}$ is an integer, multiple of the frame duration; in fact, scheduling of SUs transmissions
in the data communication phase is done immediately after spectrum sensing, hence a spectrum estimate with non-integer delay  $\delta_i^{(L)}$ can only be used for scheduling decisions with delay $\lceil \delta_i^{(L)}\rceil$. In the special case when $\delta_i^{(L)}=0$, the estimate of cell $i$ becomes immediately available to the level-$L$ cluster head;
if $\delta_i^{(L)}=1$, it becomes available for data communication in the following frame, and so on.

\begin{figure*}[t]
\centering  
\includegraphics[width=.75\linewidth,trim = 6mm 6mm 0mm 0mm,clip=true]{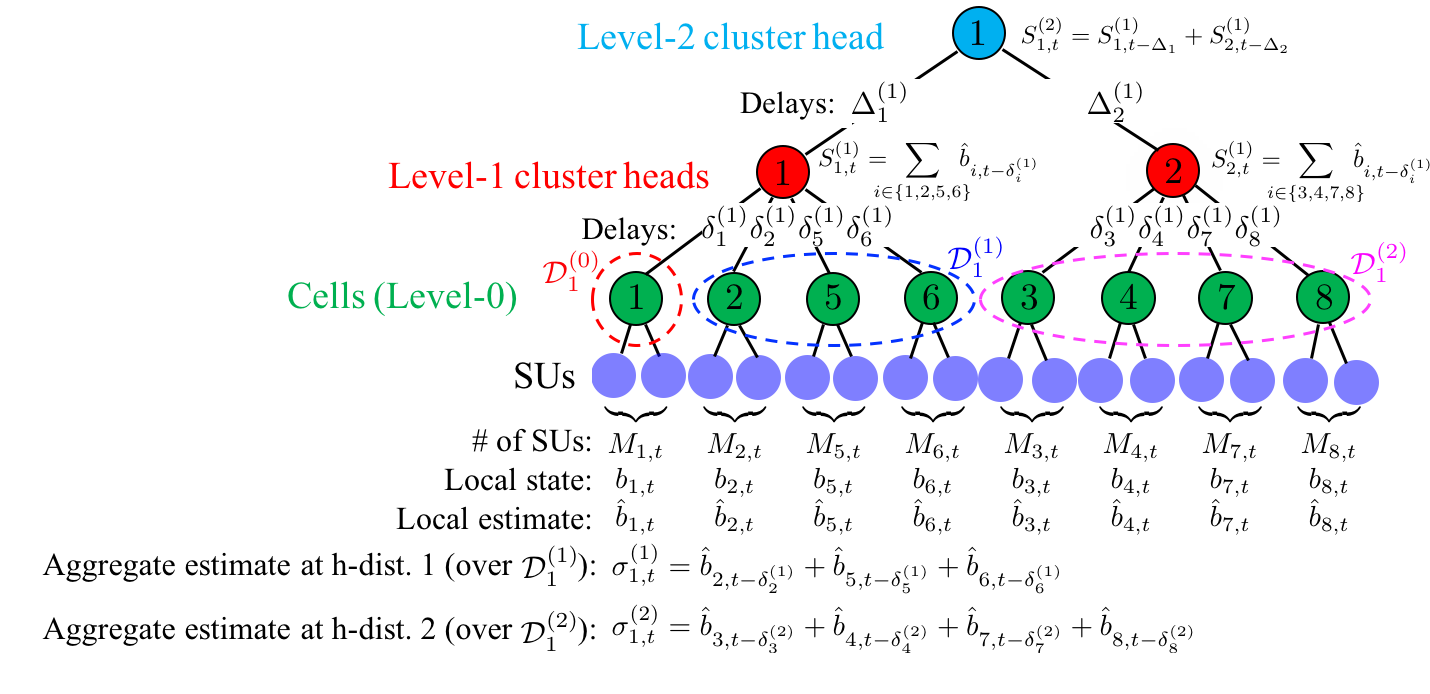}
\caption{Aggregation process referred to Fig. \ref{fig:sysmo}, and aggregate estimates relative to cell $1$.}
\label{fig:aggalgo}
\end{figure*}
 
We assume that $\delta_i^{(L)}{\leq}\delta_i^{(L+1)}$,
 \emph{i.e.}, the delay augments as the local spectrum estimates are aggregated at higher levels. 
 More precisely, let $\Delta_{m}^{(L-1)}$ be the delay between the level-$(L{-}1)$ cluster head $m$, and its level-$L$
 cluster head $n$, with $m\in\mathcal H_{n}^{(L-1)}$. We can thus express 
 $\delta_i^{(L)}$ as
 \begin{align}
 % \nonumber
 \label{deltaupdate}
 &\delta_i^{(L)}=
  \delta_i^{(L-1)}+\Delta_{h_i}^{(L-1)}
 %\nl{\nonumber\\&}
  =
 \sum_{l=1}^{L}\Delta_{h_i}^{(l-1)},
 \end{align}
 where $h_i^{(l)}$ is the level-$l$ cluster head of cell $i$.
 %\footnote{{For notational conciseness, we express $\Delta_{h_i^{(l-1)}}^{(l-1)}$ as $\Delta_{h_i}^{(l-1)}$}}
By the end of the spectrum sensing phase,
 the level-$1$ cluster head $m{\in}\mathcal H^{(1)}$ receives the spectrum estimates from
its cluster $\mathcal C_{m}^{(1)}$: $\hat b_{i,t-\delta_i^{(1)}}$
 is received from cell $i{\in}\mathcal C_m^{(1)}$ with delay $\delta_i^{(1)}{\geq}0$.
% Notice that these estimates may be received with different delays $\delta_{1,i},i\in\mathcal C_1^{(m)}$, which causes a mismatch in the In Sec. \ref{} we will show that Thus, there is a mismatch in the delays received by the 
 %Thus, a prediction step is performed before these estimates are aggregated. In particular, we define $\bar S_{i,t}^{(0)}$ as the estimate of the spectrum occupancy $b_{i,t}$ at $m\in\mathcal H_{1}$, given the delayed estimate $\hat b_{i,t-\delta_{1,i}}$. This is given by \begin{align}\bar S_{i,t}^{(p,0)}\triangleq\mathbb P(b_{i,t}=1|\hat b_{i,t-\delta_{1,i}}) =\pi_B+\mu^{\delta_{1,i}} \left(\hat b_{i,t-\delta_{1,i}}-\pi_B\right), \end{align} which corresponds to the $\delta_{1,i}$ steps transition probability from state $b_{i,t-\delta_{1,i}}\sim\mathcal B(\hat b_{i,t-\delta_{1,i}})$ to $b_{i,t}=1$. Note that, in the special case $\delta_{1,i}=0$, we obtain $\bar S_{i,t}^{(p,0)}=\hat b_{i,t}$. \nm{Define $\mathcal B$} After this prediction step, 
 These are aggregated at the level-$1$ cluster head as
\begin{align}
\label{smtp1}
S_{m,t}^{(1)}
\triangleq\sum_{i\in\mathcal C_{m}^{(1)}}
\hat b_{i,t-\delta_i^{(1)}},\ \forall m\in\mathcal H^{(1)},
\end{align}
each with its own delay.
This process continues up the hierarchy: the level-$L$ cluster head 
$m{\in}\mathcal H^{(L)}$ receives
 %\emph{delay mismatched aggregate} spectrum estimate
 $S_{k,t-\delta}^{(L-1)}$ from
the level-$(L{-}1)$ cluster heads $k{\in}\mathcal H_{m}^{(L-1)}$ connected to it, with delay $\Delta_{k}^{(L-1)}$, and aggregates them as
\begin{align}
\label{aggregationL}
S_{m,t}^{(L)}
=
\sum_{k\in\mathcal H_{m}^{(L-1)}}
S_{k,t-\Delta_k}^{(L-1)},
\end{align}
each with its own delay $\Delta_k^{(L-1)}$.
Importantly, these delays may differ from each other,
hence $S_{m,t}^{(L)}$ does not truly reflect the aggregate spectrum at a given time.
For this reason we denote $S_{m,t}^{(L)}$ as the \emph{delay mismatched aggregate spectrum estimate} at level-$L$ cluster head $m$.
The next lemma relates $S_{m,t}^{(L)}$ to the local estimates.
 \begin{lemma}
 \label{Lem:smt}
 Let $m\in\mathcal H^{(L)}$ be a level-$L$ cluster head.
Then,
 \begin{align}
 \label{sdfhb}
S_{m,t}^{(L)}
=
\sum_{j\in\mathcal C_{m}^{(L)}}\hat b_{j,t-\delta_j^{(L)}}.
\end{align}
 \end{lemma}
 \begin{proof}
 See Appendix A.
 \end{proof}

Despite mismatched delays, in Sec. \ref{analysis} we show that cell $i$ 
can compensate them via prediction. 
\begin{remark}
\label{remdec}
Note that the aggregation process runs in a decentralized fashion at each node: level-$L$ cluster head $m$ needs only information about the set of level-$(L{-}1)$ cluster heads connected to it, $k{\in}\mathcal H_{m}^{(L{-}1)}$, 
and the  delays $\Delta_k^{(L{-}1)}$. 
This information is available at each node during tree formation; delays may be estimated using time-stamps associated with the control packets.
{The aggregation process has low complexity: each cluster-head simply 
aggregates the delay mismatched aggregate spectrum estimates from the lower level cluster heads connected to it, and
 transmits this aggregate estimate to its higher level cluster head.}
\end{remark}

Eventually, the aggregate spectrum measurements are fused at the root (level-$D$) as
\begin{align}
S_{1,t}^{(D)}
=\sum_{k\in\mathcal H_{1}^{(D-1)}} 
S_{k,t-\Delta_k}^{(D-1)}
=
\sum_{j\in\mathcal C}\hat b_{j,t-\delta_j^{(D)}},
\end{align}
where we used Lemma \ref{Lem:smt} and  $\mathcal C_{1}^{(D)}{\equiv}\mathcal C$.
Upon reaching level-$D$ and each of the lower levels,
the aggregate spectrum estimates are propagated down to the individual cells $i{\in}\mathcal C$ over the tree.\footnote{We include the propagation delay from the cluster head back to the single cells
in $\delta_i^{(L)}$.}

Therefore, at the beginning of frame $t$, the SUs  in cell $i$ receive
the delay mismatched aggregate spectrum estimates
  from their level-$L$ cluster heads $h_i^{(L)},L=0,\dots,D$,
\begin{align*}
\left\{\begin{array}{rcl}S_{i,t}^{(0)}&=&\hat b_{i,t},\\S_{h_i,t}^{(L)}&=&\sum_{j\in\mathcal C_{h_i}^{(L)}}\hat b_{j,t-\delta_{j}^{(L)}},\ 1\leq L<D,\end{array}\right.\end{align*}
%$$S_{h_i,t}^{(L)}=\sum_{j\in\mathcal C_{h_i}^{(L)}}\hat b_{j,t-\delta_{j}^{(L)}},\ 0\leq L<D,$$
where we remind that %$H_L(i)$ is the level-$L$ parent of cell $i$, 
$\mathcal C_{h_i}^{(L)}$ is the set of cells associated to $h_i^{(L)}$ at level $L$,
and $\delta_j^{(L)}$ is the delay for the
estimate of $b_{j,t}$ to propagate to the level-$L$ cluster head $h_i^{(L)}$.
From this set of measurements,
cell $i$ can compute
the aggregate spectrum estimate of the cells at all h-distances from itself as
\begin{align}
\label{sigmadef}
\left\{
\begin{array}{lcl}
\sigma_{i,t}^{(0)}&\triangleq&S_{h_i,t}^{(0)}=\hat b_{i,t},\\
%&\vdots&\\
\sigma_{i,t}^{(L)}&\triangleq& S_{h_i,t}^{(L)}{-}S_{h_i,t-\Delta_{h_i}}^{(L-1)},\ 1{\leq}L{\leq}D.
\end{array}\right.
\end{align}
To interpret $\sigma_{i,t}^{(L)}$
as the aggregate estimate at h-distance $L$ from cell $i$,
note that Lemma \ref{Lem:smt} yields
\\\centerline{$
\sigma_{i,t}^{(L)}=
\sum_{j\in\mathcal C_{h_i}^{(L)}}\hat b_{j,t-\delta_j^{(L)}}
-\sum_{j\in\mathcal C_{h_i}^{(L-1)}}\hat b_{j,t-\delta_j^{(L-1)}-\Delta_{h_i}^{(L-1)}}.
$}\\
Since $i,j{\in}\mathcal C_{h_i}^{(L-1)}$ share
 the same level-$(L{-}1)$ and -$L$ cluster heads, $h_i^{(L-1)}$ and $h_i^{(L)}$,
(\ref{deltaupdate}) yields
\\\centerline{$
\delta_j^{(L-1)}+\Delta_{h_i}^{(L-1)}
%\\&
=\delta_j^{(L)}.
$}\\
Then,
$\forall\ L=1,2,\dots,D$, using Definition \ref{def2} we obtain
\begin{align}
\label{sigmadel}
&\sigma_{i,t}^{(L)}{=}
\sum_{j\in\mathcal C_{h_i}^{(L)}}\hat b_{j,t-\delta_j^{(L)}}
-\sum_{j\in\mathcal C_{h_i}^{(L-1)}}\hat b_{j,t-\delta_j^{(L)}}
%\\&
{=}
\sum_{j\in\mathcal D_i^{(L)}}\hat b_{j,t-\delta_j^{(L)}},
\end{align}
so that $\sigma_{i,t}^{(L)}$ represents the \emph{delay mismatched aggregate}
 spectrum estimate of cells at h-distance $L$ from cell $i$ ($j{\in}\mathcal D_i^{(L)}$).
Thus, with this method,
  the SUs  in cell $i$  can compute the \emph{delay mismatched aggregate}
 estimate at multiple scales corresponding to different h-distances, given delayed measurements.
Notably, only aggregate and delayed estimates are available, rather than timely information on the state of each cell.
These are used to update the belief $\pi_{i,t}$ in Sec. \ref{analysis}.

\section{Analysis}
 \label{analysis}
 Given past and current
 delayed spectrum estimates across all h-distances,
 $\boldsymbol{\sigma}_{i,\tau}
{=}
(\sigma_{i,\tau}^{(0)},
\sigma_{i,\tau}^{(1)},\dots,
\sigma_{i,\tau}^{(D)}),$ $\tau{=}0,\dots,t$,
 the form of the local belief $\pi_{i,t}$ is provided in the following theorem.
\begin{thm}
\label{thm1} 
Given $\boldsymbol{\sigma}_{i,\tau},\tau=0,1,\dots,t$,
 we have
\begin{align}
\label{eq1}
&\pi_{i,t}(\mathbf b){=}\prod_{L=0}^{D}\mathbb P\left(b_{j,t}{=}b_j,\forall j\in
\mathcal D_i^{(L)}|
\sigma_{i,\tau}^{(L)},\forall\tau=0,\dots,t
\right),
\end{align}
where, letting
$x=\sum_{j\in\mathcal D_i^{(L)}}b_j$,
\begin{multline}
\label{eq2}
\mathbb P\Bigr(b_{j,t}=b_j,\forall j\in\mathcal D_i^{(L)}\Bigr|
\sigma_{i,\tau}^{(L)},\forall\tau=0,\dots,t
\Bigr)
\\
=
\sum_{x=0}^{|\mathcal D_i^{(L)}|}
\underbrace{
\mathbb P\Bigr(
\sum_{j\in\mathcal D_i^{(L)}}
b_{j,t-\delta_j^{(L)}}
=x\Bigr|
\sigma_{i,\tau}^{(L)},\forall\tau=0,\dots,t
\Bigr)}_{\text{A}}
\\
\times\underbrace{\frac{x!|\mathcal D_i^{(L)}-x|!}{|\mathcal D_i^{(L)}|!}}_{\text{B}}
\sum_{\tilde b_{j},j\in\mathcal D_i^{(L)}}
\underbrace{\chi\Bigr(\sum_{l\in\mathcal D_i^{(L)}}\tilde b_{l}=x\Bigr)}_{\text{C}}
 \\
\times\prod_{l\in\mathcal D_i^{(L)}}
\underbrace{\left[\pi_B{+}\mu^{\delta_l^{(L)}}\left(\tilde b_{l}{-}\pi_B\right)\right]^{b_l}}_{\text{D}}
\underbrace{\left[1{-}\pi_B{-}\mu^{\delta_l^{(L)}}\left(\tilde b_{l}{-}\pi_B\right)\right]^{1{-}b_l}}_{\text{E}},
\end{multline}
where $\chi(\cdot)$ is the indicator function.
Additionally,
\begin{align}
\label{eq2_2}
%\sum_{x=0}^{|\mathcal D_i^{(L)}|}
%x\mathbb P\left(\left.\sum_{j\in\mathcal D_i^{(L)}}b_{j,t-\delta_j^{(L)}}{=}x\right|\sigma_{i,\tau}^{(L)},\forall\tau=0,\dots,t\right)
\mathbb E\Bigr(\sum_{j\in\mathcal D_i^{(L)}}b_{j,t-\delta_j^{(L)}}\Bigr|\sigma_{i,\tau}^{(L)},\forall\tau=0,\dots,t\Bigr)
=\sigma_{i,t}^{(L)}.\end{align}
\hfill
\end{thm}
\begin{proof}
See Appendix B.
\end{proof}
We note the following facts related to Theorem \ref{thm1}:
\begin{enumerate}
\item Equation (\ref{eq1}) implies that $\pi_{i,t}$ is statistically independent 
across the subsets of cells at different h-distances from cell $i$; this result follows from Assumption \ref{assum1},
which guarantees independence of spectrum occupancies and spectrum sensing across cells.
\item 
Equation (\ref{eq2}) contains five terms. "$A$" is the probability distribution of the delay mismatched aggregate spectrum occupancy given past estimates. "$B$" is the probability of a specific realization of $b_{j,t-\delta_j^{(L)}},j{\in}\mathcal D_i^{(L)}$, given that its aggregate equals $x$,
whereas "C" is the marginal over all these realizations;
since there are $|\mathcal D_i^{(L)}|!/x!/(|\mathcal D_i^{(L)}|{-}x)!$ combinations of such spectrum occupancies,
Assumption~\ref{assum1}
implies that they are uniformly distributed, yielding "$B$".\footnote{If Assumption~\ref{assum1} does not hold, estimates of aggregate occupancies could provide information as to
favor certain realizations over others, for instance, by leveraging different temporal correlations at different cells.}
Finally, terms "$D$" and "$E$" represent the $\delta_l^{(L)}$ steps transition probability from $b_{l-\delta_l^{(L)}}{=}\tilde{b}_l$ to $b_{j,t}{=}1$ and $b_{j,t}{=}0$, respectively.
\item Equation (\ref{eq2_2}) 
states that the expected delay mismatched aggregate occupancy over  $\mathcal D_i^{(L)}$ equals $\sigma_{i,t}^{(L)}$, \emph{independently} of past spectrum estimates. %This is a consequence of Assumption~\ref{assum1}.
However, its probability distribution
("$A$" in (\ref{eq2}))
\emph{does} depend on past estimates.
\item In general, the
term "$A$" in (\ref{eq2}) cannot be computed in closed form, except in some special cases (\emph{e.g.}, noiseless measurements \cite{MicheICC}). However, we will now show that a closed-form expression is not required to compute  $I_{P,i}(\pi_{i,t})$, hence the expected utility in cell $i$ via  \eqref{utility}.
To this end,
in the next lemma we compute
$\mathbb P(b_{j,t}{=}1|\pi_{i,t})$ in closed form.
\end{enumerate}
 \begin{lemma}
\label{lem2}
For $j\in\mathcal D_i^{(L)}$, \emph{i.e.},
we have
\begin{align}
\label{probexp}
&\mathbb P(b_{j,t}=1|\pi_{i,t})
=
\pi_B+\mu^{\delta_j^{(L)}}
\Biggr(\frac{\sigma_{i,t}^{(L)}}
{|\mathcal D_i^{(L)}|}-\pi_B\Biggr).
\end{align}
\end{lemma}
\begin{proof}
See Appendix C.
\end{proof}
We now compute $I_{P,i}(\pi_{i,t})$. Partitioning
  $\mathcal C$ based on the h-distances from $i$, \eqref{XXX} yields
\begin{align}
\label{XXX2}
I_{P,i}(\pi_{i,t})
\triangleq 
\sum_{L=0}^{D}\sum_{j\in\mathcal D_i^{(L)}}
\frac{\phi_{j,i}}{\phi_{i,i}}\mathbb P\left(b_{j,t}=1|\pi_{i,t}\right).
\end{align}
Then,  substituting \eqref{probexp} in (\ref{XXX2}) and letting
 \begin{align}
\label{delaymismatinter}
\left\{
\begin{array}{l}
\Phi_{\mathrm{tot},i}\triangleq
\sum_{j\in\mathcal C}\frac{\phi_{j,i}}{\phi_{i,i}},
\\
\Phi_{\mathrm{del},i}^{(L)}\triangleq\sum_{j\in\mathcal D_i^{(L)}}\mu^{\delta_j^{(L)}}
\frac{\phi_{j,i}}{\phi_{i,i}}
\end{array}
\right.
\end{align}
be the total \emph{mutual} interference generated between the SUs in cell $i$ and the PU network ($\Phi_{\mathrm{tot},i}$),
and the \emph{delay compensated mutual interference} generated between cell $i$ and the
cells at h-distance $L$ from cell $i$ ($\Phi_{\mathrm{del},i}^{(L)}$), we obtain the following lemma.
\begin{lemma}
\label{lem:exprew}
{The expected PU activity experienced in cell $i$ is given by
\begin{align}
\label{XXX3}
I_{P,i}(\boldsymbol{\sigma}_{i,t})
\triangleq 
\pi_B\Phi_{\mathrm{tot},i}
+\sum_{L=0}^{D}
\Biggr(\frac{\sigma_{i,t}^{(L)}}{|\mathcal D_i^{(L)}|}-\pi_B\Biggr)
\Phi_{\mathrm{del},i}^{(L)}.
\end{align}}
\end{lemma}
Above, for convenience, we have expressed the dependence of $I_{P,i}(\cdot)$ on 
$\boldsymbol{\sigma}_{i,t}$, rather than on $\pi_{i,t}$.
Thus, the local utility (\ref{utility}) can be computed accordingly.
{Note that $I_{P,i}(\boldsymbol{\sigma}_{i,t})$ depends on the clustering of cells
across multiple spatial scales that affect the delay mismatched aggregate
spectrum estimates $\sigma_{i,t}^{(L)}$, hence on the tree employed for hierarchical information exchange.
 In the next section, we propose a tree design matched to the structure of interference.}

\section{Tree Design}
 \label{treedesign}
The network utility depends crucially on the tree employed for information exchange. Its optimization over all possible trees is a combinatorial problem with high complexity. Thus, we use {\em agglomerative} clustering, developed in \cite[Ch. 14]{friedman:01}, in which a tree is built by successively combining smaller clusters based on a "closeness" metric, that we now develop.

Note that in our problem the goal is for cell $i$ to estimate 
the INR generated  to the PUs as accurately as possible, $\sum_{j=1}^{N_C}\frac{\phi_{j,i}}{\phi_{i,i}}b_{j,t}$.
This estimate is denoted as $I_{P,i}(\pi_{i,t})$, see \eqref{XXX}.
In fact, given $I_{P,i}(\pi_{i,t})$, SUs in cell $i$ can schedule the optimal SU traffic 
$a_{i,t}^*(I_{P,i}(\pi_{i,t}))$ via \eqref{rewardinb}, hence the optimal utility via \eqref{utility}.
With the hierarchical information exchange described in the previous section,
this estimate is given by \eqref{XXX3}.

\label{p19}
Therefore, the goal is to design the tree in such a way as to estimate 
$\sum_{j=1}^{N_C}\frac{\phi_{j,i}}{\phi_{i,i}}b_{j,t}$ as accurately as possible via 
$I_{P,i}(\boldsymbol{\sigma}_{i,t})$ in \eqref{XXX3}. 
At the same time, since all cells share the same tree, such design should 
take into account this goal across all cells.
We develop a heuristic metric to attain this goal.
To this end, we notice the following facts:
1) since higher levels correspond to larger and larger clusters over which spectrum estimates are aggregated
(for instance, with reference to Fig.~\ref{fig:sysmo}, 
$|\mathcal D_1^{(0)}|{=}1$,
$|\mathcal D_1^{(1)}|{=}3$ and $|\mathcal D_1^{(2)}|{=}4$ at h-distances $0,1,2$, respectively),
higher levels correspond to coarser estimates of spectrum occupancy, whereas lower levels
correspond to fine-grained estimates; 2) from
\eqref{XXX3}, it is apparent that 
terms with larger $\Phi_{\mathrm{del},i}^{(L)}$ affect more strongly $I_{P,i}(\boldsymbol{\sigma}_{i,t})$.
Therefore,
cluster aggregation resulting in larger $\Phi_{\mathrm{del},i}^{(L)}$
should occur at lower hierarchical levels, associated with fine-grained estimation. 
Taking these facts into account, we denote the "aggregation" metric between $n,m{\in}\mathcal H^{(L)}$ as
\begin{align}
\label{similarity}
\Gamma_{n,m}^{(L)}{=}
 \mu^{\Delta_{n,m}}
 \Biggr[
\sum_{i\in\mathcal C_n^{(L)}}
 \sum_{j\in\mathcal C_m^{(L)}}
 \mu^{\delta_j^{(L)}}\frac{\phi_{j,i}}{\phi_{i,i}}
 {+}\sum_{i\in\mathcal C_m^{(L)}}
 \sum_{j\in\mathcal C_n^{(L)}}
 \mu^{\delta_j^{(L)}}\frac{\phi_{j,i}}{\phi_{i,i}}
 \Biggr].
\end{align}
$\Gamma_{n,m}^{(L)}$ represents the benefit of aggregating together the clusters associated to level-$L$ cluster-heads $m$ and $n$,
$\mathcal C_m^{(L)}$ and $\mathcal C_n^{(L)}$, respectively, into one level-$(L{+}1)$ cluster, and 
 $\Delta_{n,m}$ is the additional delay incurred to aggregate them.\footnote{$\Delta_{n,m}$ can be chosen, for instance, based on the number of hops traversed to aggregate estimates at the upper level $(L+1)$. This number is approximately proportional to the distance between cluster heads $n$ and $m$.}
In fact, if such aggregation occurs,
from the perspective of cell $i\in\mathcal C_n^{(L)}$,
$\mathcal C_m^{(L)}$ will become the set of cells at h-distance $L+1$ from cell $i$,
$\mathcal D_i^{(L+1)}\equiv\mathcal C_m^{(L)}$, so that,
letting $\delta_j^{(L+1)}{=}\Delta_{n,m}{+}\delta_j^{(L)}$ as in \eqref{deltaupdate},
 the first term associated to $i$ in \eqref{similarity} is equivalent to
\begin{align}
\sum_{j\in\mathcal C_m^{(L)}}\mu^{\Delta_{n,m}+\delta_j^{(L)}}\frac{\phi_{j,i}}{\phi_{i,i}}
=\Phi_{\mathrm{del},i}^{(L+1)}.
\end{align}
The second term in \eqref{similarity}
 has a similar interpretation, relative to cell $i\in\mathcal C_m^{(L)}$. Thus, the aggregation metric
 $\Gamma_{n,m}^{(L)}$ corresponds to
 $\sum_{i\in\mathcal C_n^{(L)}}\Phi_{\mathrm{del},i}^{(L+1)}+\sum_{i\in\mathcal C_m^{(L)}}\Phi_{\mathrm{del},i}^{(L+1)}$,
 if clusters $\mathcal C_n^{(L)}$ and $\mathcal C_m^{(L)}$ are aggregated together.
As justified previously, this quantity should be made as large as possible in order to maximize the informativeness of the 
aggregation of estimates.

 In addition, we want to limit the cost incurred to send measurements up and down the hierarchy.
Assuming that estimates are transmitted via multi-hop, the cost will be proportional to the distance between clusters. Thus, each time we combine two clusters $\mathcal C_n^{(L)}$ and $\mathcal C_m^{(L)}$ to form the tree, we incur an additional {\em aggregation cost per cell} $C_{n,m}$, 
defined as
\begin{equation}
\label{worstcasecost}
	C_{n,m}\ =  \frac{1}{N_C}\max_{i \in \mathcal{C}_n^{(L)},j \in \mathcal{C}_m^{(L)}}d_{i,j},
\end{equation}
representing the worst-case aggregation cost, where $d_{i,j}$ is the distance between cells $i$ and $j$.

\begin{algorithm}%[t]
%\begin{multicols}{2}
\small
\SetKwData{Left}{left}\SetKwData{This}{this}\SetKwData{Up}{up}
\SetKwFunction{Union}{Union}\SetKwFunction{FindCompress}{FindCompress}
\SetKwInOut{Input}{input}\SetKwInOut{Output}{output}
\Input{Cells $\mathcal{C}$, interference matrix $\boldsymbol{\Phi}$, max cost $C_{\max}$ (per cell)}
\Output{A hierarchy of clusters $\mathcal{C}_k^{(L)}$, $k\in\mathcal H^{(L)}$, $L=1,\dots, D$, delays $\delta_i^{(L)}$, and aggregation cost $C_{\mathrm{cell}}$}
\BlankLine
{\bf Initialize:} $L{\leftarrow}0$, $\mathcal H^{(L)}{\leftarrow}\mathcal C$, $\mathcal{C}_i^{(0)}{\leftarrow}\{i\}$, $\delta_i^{(0)}{=}0,\forall i{\in}\mathcal C$, $C_{\mathrm{cell}}{=}0$\;
\BlankLine
\Repeat{{\bf termination}}{
%    \tcp{compute delays/cost (e.g., $\propto$\#hops)}
  $\Delta_{n,m},C_{n,m},\forall n,m\in\mathcal{H}^{(L)},n\neq m$
  (delays and cost are computed, e.g., $\propto$\#hops)\;
%      \tcp{list of unpaired feasible pairs}
    $\mathcal F^{(L)}{\leftarrow}\{
    (n,m){\in}\mathcal{H}^{(L)2}:n,{\neq}m,
 C_{\mathrm{cell}}{+}C_{n,m}{\leq}C_{\max}
    \}$
   (set of unpaired feasible pairs)\;
  \If{
  $|\mathcal F^{(L)}|=0$ \emph{(cost exceeded)}
  }{
%  \tcp{cost exceeded}
  \bf{terminate}}
% \& $\exists (k_1,k_2)\in\mathcal H_L^2$ s.t. $C_{\mathrm{tot}}+c_L(k_1,k_2)\leq C_{\max}$
% \tcp{empty list of next level cluster heads and cluster head counter} 
	$\mathcal{H}^{(L+1)}\leftarrow \emptyset$,
	$k_{next}\leftarrow 1$ (empty set of next level cluster heads and cluster head counter)\;  
	%\tcp{list of unpaired cluster heads} 
	$\mathcal{H}^{(L)}_{unp} \leftarrow \mathcal{H}^{(L)}$ (set of unpaired cluster heads)\; 
\While{$|\mathcal F^{(L)}|>0$}{
		%\tcp*[l]{no unpaired neighbors}
	%\tcp{find unpaired feasible cluster pair with max $\Gamma$}
	$(n^*,m^*){\leftarrow}
    \underset{(n,m){\in}\mathcal{F}^{(L)}}{\arg\max} \Gamma_{n,m}^{(L)}$ 
    (find unpaired feasible cluster pair with max $\Gamma$, see (\ref{similarity}))\;
	$\mathcal{H}^{(L+1)} \leftarrow \mathcal{H}^{(L+1)}\cup\{k_{next}\},\ 
		\mathcal{C}_{k_{next}}^{(L+1)} \leftarrow \mathcal{C}_{n^*}^{(L)}\cup\mathcal{C}_{m^*}^{(L)}$\; 
        %\tcp{update delay/cost}
  $\delta_i^{(L+1)}{=}\delta_i^{(L)}{+}\Delta_{n^*,m^*},\forall i{\in}\mathcal{C}_{k_{next}}^{(L+1)},
  C_{\mathrm{cell}}{\leftarrow}C_{\mathrm{cell}}{+}C_{n^*,m^*}
  $ (update delay and cost)\;
  			%\tcp{remove paired clusters}
	$\mathcal{H}_{unp}^{(L)} \leftarrow \mathcal{H}_{unp}^{(L)}\setminus \{n^*,m^*\}$ (remove paired clusters)\;
  %\tcp{updated feasible pairs}
    $\mathcal F^{(L)}\leftarrow\{
    (n,m){\in}\mathcal{H}_{unp}^{(L)}\times\mathcal{H}_{unp}^{(L)}:n,{\neq}m,
 C_{\mathrm{cell}}{+}C_{n,m}{\leq}C_{\max}
    \}$  (updated feasible pairs)\;
		$k_{next} \leftarrow k_{next}+1$ \; 
    }
%        		\tcp{unpaired clusters incur excessive cost, ``pair'' each with itself}
	\ForAll{$k\in\mathcal{H}_{unp}^{(L)}$ \emph{(unpaired clusters incur excessive cost, ``pair'' each with itself)}}{
$\mathcal{H}^{(L+1)} \leftarrow \mathcal{H}^{(L+1)}\cup\{k_{next}\},\ 
\mathcal{C}_{k_{next}}^{(L+1)} \leftarrow \mathcal{C}_{k}^{(L)}$ \; 
%        \tcp{no additional delay/cost}
        $\delta_i^{(L+1)}=\delta_i^{(L)},\forall i\in\mathcal{C}_{k_{next}}^{(L+1)}$ (no additional delay/cost)\; 
		$k_{next} \leftarrow k_{next}+1$ \; 
}
%\tcp{Proceed to the next level}
	$L \leftarrow L+1$ (Proceed to the next level)\;
}
%}
\caption{
Hierarchical Aggregation Tree Construction}\label{alg:clustering}
%\end{multicols}
\end{algorithm}

The algorithm proceeds as shown in Algorithm \ref{alg:clustering}. We initialize it with the $N_C$ sets containing the single cells, $\mathcal{C}_i^{(0)}=\{i\},i=1,2,\dots,N_C$, and aggregation cost (per cell) $C_{\mathrm{cell}}= 0$. Then, at each level-$L$, we iterate over all cluster pairs, pairing those with highest  aggregation metric $\Gamma$. This forms the set of level-$(L+1)$ clusters; we update the delays accordingly and update $C_{\mathrm{cell}}$ by adding $C_{n^*,m^*}^{(L)}$. If the number of clusters at level-$L$ happens to be odd, one cluster may not be paired, in which case it forms its own level-$(L{+}1)$ cluster, and the delay remains unchanged. The algorithm proceeds until either: (1) the cluster $\mathcal{C}_{1}^{(L)}$ contains the entire network, \emph{i.e.}, a tree has been formed, or (2) $C_{\mathrm{cell}}> C_{\mathrm{max}}$, i.e., the allowed cost is exceeded.  %If the second condition occurs, the resulting ``tree'' has no root.
Agglomerative clustering has complexity $O(N_C^2 \log(N_C))$, where the term $N_C^2$ owes to searching over all pairs of clusters, and the term $\log(N_C)$ is related to the tree depth, which is logarithmic in the number of cells  \cite[Ch. 14]{friedman:01}.
In the next section, we will compare our scheme with the consensus-based scheme \cite{li:TVT10}:
 this scheme requires a "connected" graph to achieve consensus, whose complexity is 
 $O(N_C^3d)$, with $d$ being the desired degree of each node in the graph \cite{Bhuiyan}.
 Therefore, by leveraging the tree structure, our tree construction is more computationally efficient. However, tree design 
 will be executed only at initialization, or when the network topology changes, which is infrequent in fixed cellular networks as considered in this work, hence it is not expected to have a significant impact on the long-term performance.
 \section{Numerical Results}
  \label{numres}
 In this section, we provide numerical results based on Monte Carlo simulations.
We adopt a model with stochastic blockage \cite{bai:TWC2014}: rectangular blockages of fixed height and width are placed randomly on the {\em boundaries} between cells. 
   Each blockage has width $1$ and height $5$, and is randomly placed.
We say that links between cells $i,j$ are {\em line of sight} (LOS) if the line segment connecting the centers of cells $i$ and $j$ does not intersect any blockage object. Otherwise, such links are said to be {\em non-LOS} (NLOS). Accordingly, we define LOS and NLOS large-scale pathloss exponents $\alpha_L{=}2.1$ and $\alpha_N{=}3.3$, respectively.
These values were derived experimentally in \cite[Table I]{Sun} at a reference frequency of $2\mathrm{GHz}$. 
 %Experimentally-derived values are published in \cite{rappaport,rappaport:TAP2013}, with $\alpha_L \approx 2$ and $\alpha_N \approx 3-4$ as typical values.

  \begin{figure*}
    \centering
    \hspace{5mm}
    \subfigure[Impact of blockages, delay parameter $\gamma=0$.]
    {
        \includegraphics[width=0.4\linewidth,trim = 6mm 0mm 15mm 8mm,clip=true]{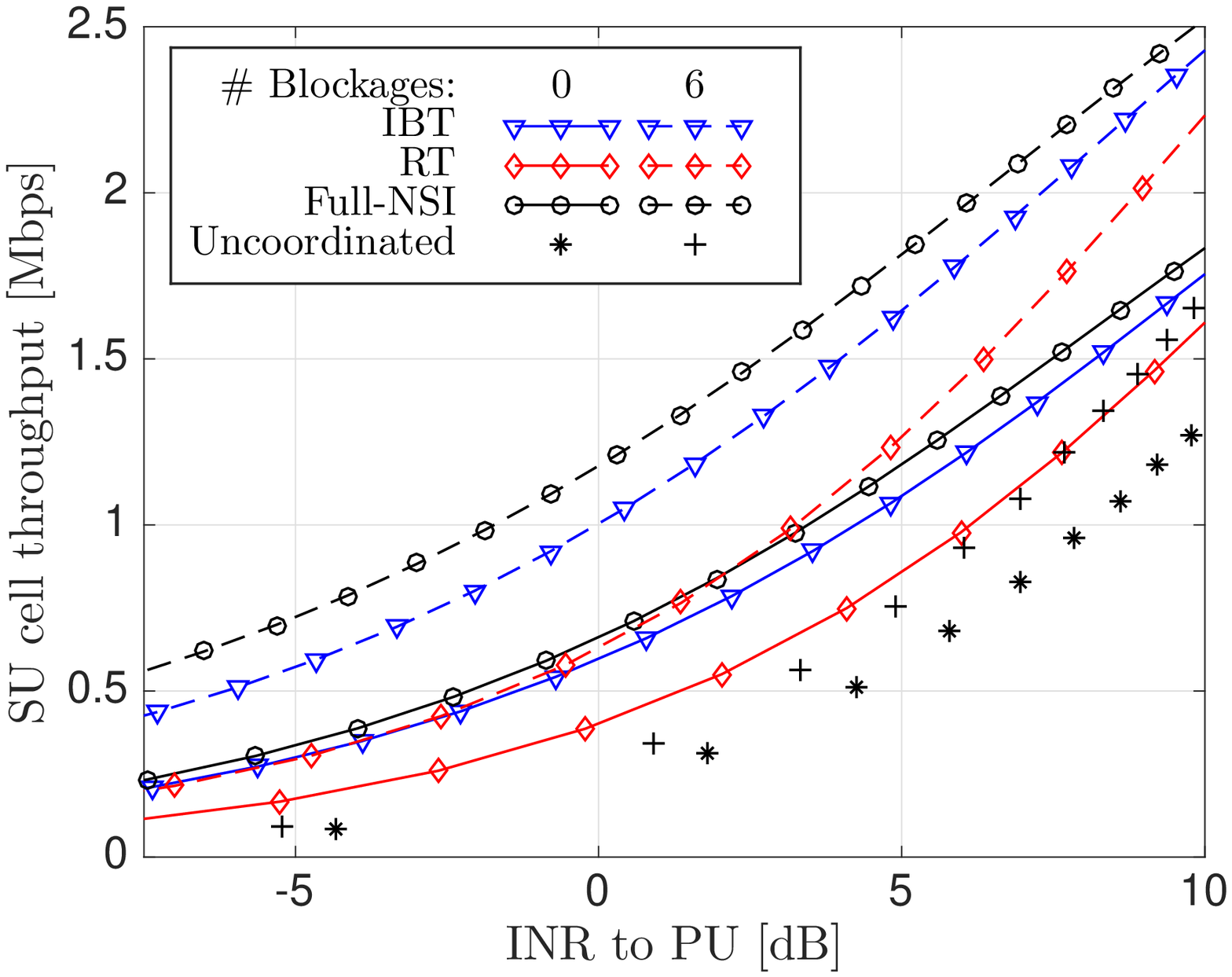}
        \label{fig:simres1}
    }
    \hfill
    \subfigure[Impact of delay, $1$ blockage.]
    {
        \includegraphics[width=0.4\linewidth,trim = 6mm 0mm 15mm 8mm,clip=true]{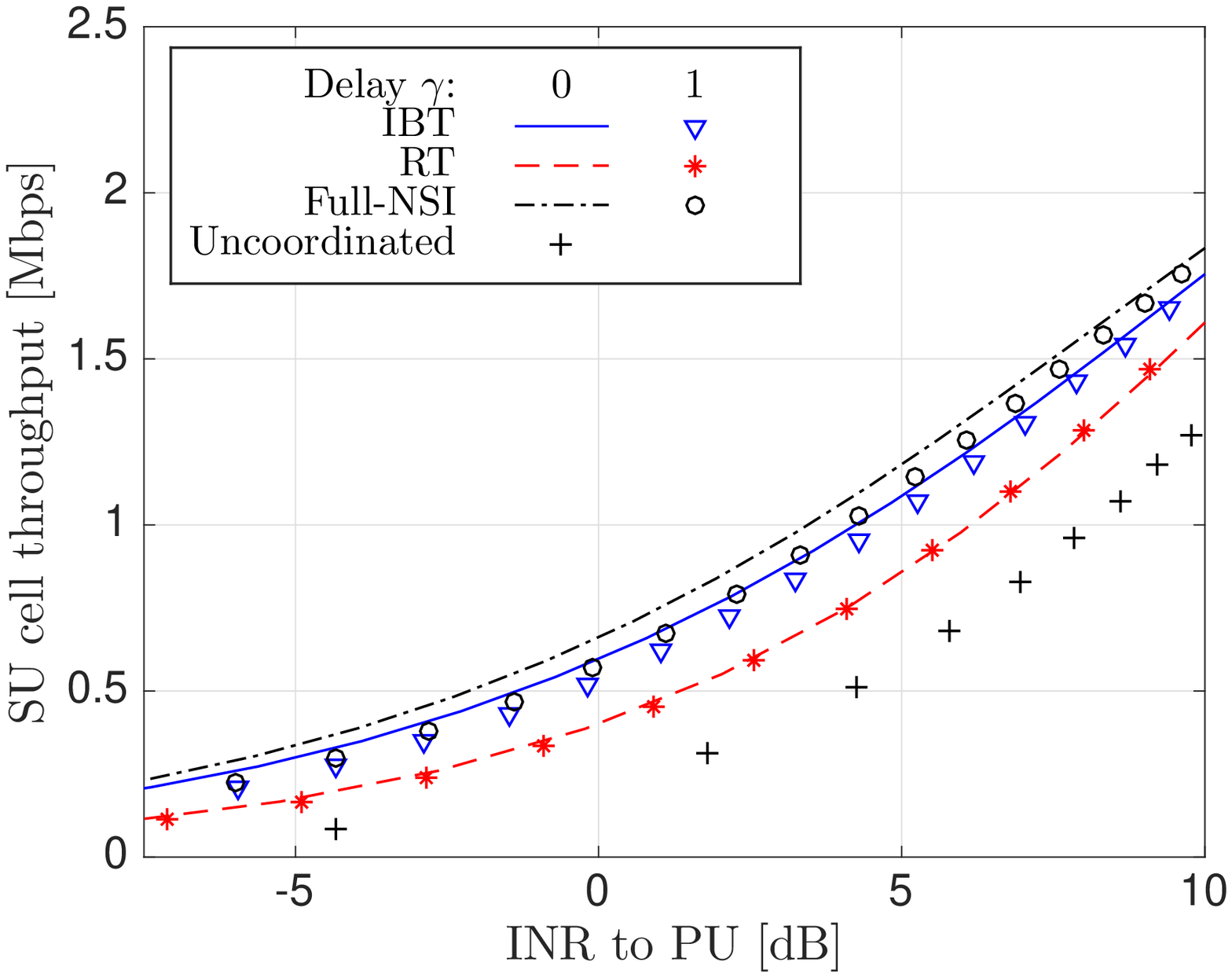}
        \label{fig:simres2}
    }
        \hspace{5mm}
    \caption{SU cell throughput versus average INR experienced at PUs, cost constraint $C_{\max}=\infty$.}
    \label{fig:simres12}
\end{figure*}
  
 In the simulations,
 we consider a $16{\times}16$ cells network over an area of $1.6\mathrm{\mathrm km}{\times}1.6\mathrm{\mathrm km}$. We set the parameters as follows: 
SINR decoding threshold $\mathrm{SINR}_{\mathrm{th}}{=}5\mathrm{dB}$,
noise power spectral density $N_0{=}-173\mathrm{dBm/Hz}$, bandwidth $W_{tot}{=}20\mathrm{MHz}$,
   $\nu_1{=}0.005$, $\nu_0{=}0.095$, hence $\pi_B{=}0.05$ and $\mu{=}0.9$. The interference matrix $\boldsymbol{\Phi}$
  is calculated as  in \eqref{pathloss}, where
 $P_{tx}{=}{-}11\mathrm{dBm}$ is the transmission power, common to all PUs and SUs,
 $L_{ref}{=}74\mathrm{dB}$ is the large-scale pathloss based on Friis' free space propagation, calculated at a reference distance $d_{ref}{=}50\mathrm{m}$ (equal to the average cell radius);
 $\alpha_{i,j}{=}\alpha_L$ if there is LOS between the centers of cells $i$ and $j$,
 otherwise, $\alpha_{i,j}{=}\alpha_N$ in case of NLOS (path obstructed by blockage).
 
  We assume that local estimation is error-free ($\epsilon_{F}{=}\epsilon_M{=}0$) and $M_{i,t}{\gg}1,\forall i,t$,
  corresponding to a dense setup with large number of SUs.
  In this work, we do not consider the overhead of local spectrum sensing within each cell,
  which can be severe in dense networks and may be reduced by using decentralized techniques to select the most informative SUs, such as  in  \cite{Wu2012}; these considerations are outside the scope of this paper, and are left for future work.
  We average the results over $200$ realizations of the blockage model. For each one of these, we generate
  a sequence of $1000$ frames to generate the Markov process $\{\mathbf b_t,t\geq 0\}$.
  We consider the following schemes:

\begin{itemize}
\item a scheme with the \emph{interference-based tree} (IBT) generated with Algorithm~\ref{alg:clustering} by leveraging the specific structure of interference, delays and aggregation costs;
\item a scheme with a \emph{{random} tree} (RT),
{in which the "max $\Gamma$"  cluster association in Algorithm~\ref{alg:clustering} is replaced
with a random association. The aim of using this scheme is to test the importance of generating a tree \emph{matched} to the structure interference;}
\item a scheme with full (but delayed) NSI (Full-NSI); since this scheme represents the best we can do, provided that we can afford the cost of acquisition of full NSI, it will be used to evaluate the sub-optimality of the proposed IBT in terms of the trade-off between SU cell throughput and interference to PUs;\label{p6}
\item an uncoordinated scheme where SUs access the spectrum with constant probability $p_{tx}$, i.i.d. over time and across SUs (Uncoordinated).
\end{itemize}
 We assume that the delay to propagate spectrum measurements between cells $i$ and $j$ is proportional to their distance, \emph{i.e.}, $\delta_{i,j}=\gamma d_{i,j}$, where $\gamma$ is varied in $[0,1]$.

In order to separate the effects of blockages, delay, and cost of aggregation on the performance, we evaluate the impact of: 1) Blockages, but no delay nor cost constraint ($\gamma{=}0$, $C_{\max}{=}\infty$, Fig. \ref{fig:simres1});
2) Delay, with one blockage but no cost constraint (1 blockage, $C_{\max}{=}\infty$, Fig. \ref{fig:simres2}); 3) 
Cost of aggregation, with one blockage and no delay (1 blockages, $\gamma{=}0$, Fig. \ref{fig:simres3}).
In all these figures, unless otherwise stated, we evaluate the lower bound to the SU cell throughput, given by
 \eqref{rtt}
and the INR experienced at the PUs
(both averaged over cells and over time). 
We vary the parameter $\lambda$ in the utility function \eqref{utility}
{and the SU access probability $p_{tx}$ in the "Uncoordinated" scheme,}
 to obtain the desired trade-off between SU cell throughput and INR.
 
In Fig. \ref{fig:simres1}, we
notice that, for all schemes, the presence of blockages improves the performance.
 In fact, \emph{blockages provide a form of interference mitigation}.
 By comparing the schemes with each other,
the best performance is obtained with Full-NSI. In fact, each cell can leverage the most refined information on
the interference pattern. However, as we will see in Fig. \ref{fig:simres3}, \emph{this comes at a huge cost to propagate NSI over the network.}
Remarkably, IBT incurs only a 15\% (for 6 blockages) and 10\% (for no blockages) performance degradation with respect to Full-NSI, for a reference INR of $0$dB
(this result becomes more remarkable when comparing the aggregation costs in Fig. \ref{fig:simres3}).
Additionally, RT incurs a severe performance degradation with respect to IBT
(60\% and 30\% degradation for 6 blockages and no blockages, respectively, for a reference INR of $0$dB);
this fact highlights the importance of designing a tree matched to the structure of interference, as done in Algorithm \ref{alg:clustering}, and validates our choice of the $\Gamma$ metric used to associate clusters in the algorithm, defined in \eqref{similarity}.
Finally, we observe that the "Uncoordinated" scheme performs the worst, since it does not adapt the SU transmissions to interference.

\begin{figure}
\centering  
\includegraphics[width=0.5\linewidth,trim = 2mm 0mm 14mm 8mm,clip=true]{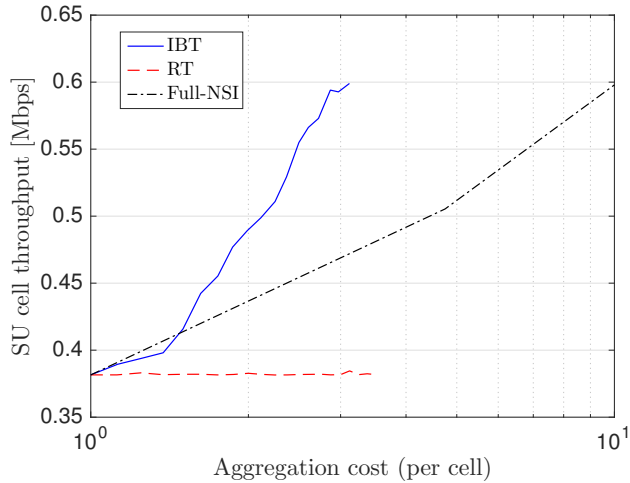}
\caption{Impact of aggregation cost on the
SU cell throughput, with $0$dB maximum constraint on the
 average INR caused to PUs. $1$ blockage; no delay.}
         \label{fig:simres3}
\end{figure}

 In Fig. \ref{fig:simres2}, we evaluate the impact of delay (note that "Uncoordinated" is not affected by delays). 
 As expected, the SU cell throughput decreases as the delay augments. This follows from the fact that 
 delayed spectrum estimates represent less accurately the actual spectrum occupancy, and may become outdated,
 and thus less informative for scheduling decisions of SUs.
 However, the performance degradation is minimal. In fact,
 the spectrum occupancy varies slowly over time: the expected duration of
 a period during which the spectrum is occupied by a PU is $1/\nu_0{\simeq}10$ frames,
 hence only the spectrum estimates received with delay larger than 10 become non informative;
 these estimates, in turn, correspond to cells that are farther away from the reference cell, hence 
 less susceptible to interference caused by the reference cell.\footnote{We remind that the delay to propagate spectrum measurements between cells $i$ and $j$ is  $\delta_{i,j}{=}\gamma d_{i,j}$, hence only farther cells are affected by large delays.}
 We notice a similar trend as in Fig.~\ref{fig:simres1} in terms of the comparison among the schemes employed.
 
  In Fig. \ref{fig:simres3}, we evaluate the trade-off between aggregation cost and performance. 
To this end:
 \begin{itemize}
 \item We vary the cost constraint $C_{\max}$ in Algorithm \ref{alg:clustering} to obtain a trade-off for IBT and RT; we use a "worst-case" cost evaluation with multi-hop, given by (\ref{worstcasecost}).
% \item We perform aggregation up to level $L_{\max}$ in the regular tree, and vary $L_{\max}=1,2,\dots,D$ to obtain a trade-off for RT; we use the same cost metric as IBT.
 \item To evaluate Full-NSI, each cell collects \emph{partial} but fine-grained NSI up to a certain radius; larger radius corresponds to more comprehensive NSI but larger cost; using
 multi-hop for NSI aggregation, the cost equals approximately the number of cells within the  radius. This scheme borrows from \cite{Vasilakos}, where
 each cell informs neighboring ones of the resource blocks used by its users. 
 \end{itemize}
We notice that IBT achieves a much better trade-off than Full-NSI: it enables SUs to gather relevant information 
for scheduling decisions, with minimal cost in the exchange of state information.
 In fact, by aggregating NSI at multiple layers, as opposed to maintaining fine-grained NSI, IBT retains the gains of partial NSI, but at a much smaller cost of aggregation.
In particular, for a reference SU cell throughput of 0.6Mbps, IBT incurs one-third of the cost of aggregation of Full-CSI.
On the other hand, RT does not improve as the cost increases; in fact, the random tree construction in RT results in information exchange which is not matched to the structure of interference, hence less informative to network control.

\begin{figure}
\centering  
\includegraphics[width=0.5\linewidth,trim = 2mm 0mm 14mm 8mm,clip=true]{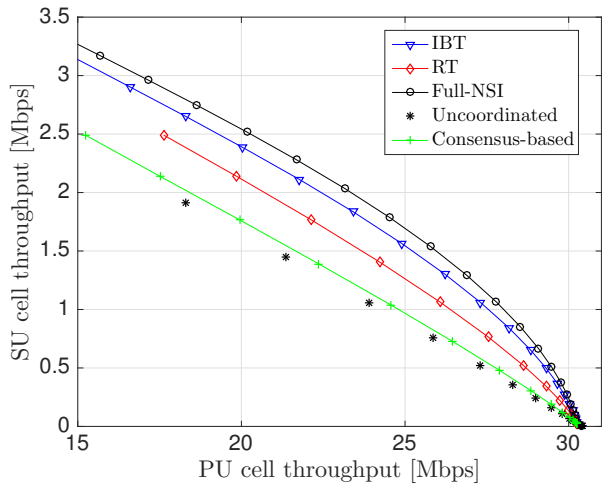}
\caption{Simulation with random topology and realistic large-scale pathloss. Comparison with "consensus" scheme \cite{li:TVT10}.}
\label{fig:simresreal}
\end{figure}

So far, in our analysis and numerical evaluation we have assumed that large-scale pathloss is calculated between cell centers, and collected in the INR matrix
$\boldsymbol{\Phi}$. However, large-scale pathloss between a transmitter and a receiver depends on their mutual position within their respective cell.
Additionally, we used the SU cell throughput lower bound \eqref{rtt}.
This motivates us to evaluate the performance in a more realistic scenario, where these assumptions are relaxed. In Fig. \ref{fig:simresreal}, we evaluate a realistic scenario with the following features:
\begin{itemize}
\item We generate $100$ independent realizations of the network topology with $N_C{=}256$ PU cells; in each realization, the transmitter-receiver pairs are deployed randomly over an area of $1.6\mathrm{\mathrm km}{\times}1.6\mathrm{\mathrm km}$; an irregular cell topology is thus defined based on minimum distance; 10 SUs are deployed randomly in each cell (each with its own receiver).
\item The large-scale pathloss is computed between each transmitter and receiver based on their relative distance, as in \eqref{pathloss}. The INR matrix $\boldsymbol{\Phi}$ is computed relative to the cell centers. This is used to construct the
hierarchical aggregation tree (Algorithm \ref{alg:clustering}), to estimate $I_{S,i}(t)$ and $I_{P,i}(\pi_{i,t})$ as in \eqref{Ts} and \eqref{XXX}, hence to compute the optimal SU traffic $a_{i,t}^*$ as in \eqref{rewardinb}. However, the performance is evaluated under the \emph{actual} distance-dependent large-scale pathloss and the realization of the Rayleigh fading process, as described in the next item.
\item For each realization of the network topology, we generate 1000 frames with random SU access decisions; the PU spectrum occupancy process $\mathbf b_t$ evolves according to the Markov process described in Sec. \ref{sysmo}, with $\nu_1{=}0.005$, $\nu_0{=}0.095$; in each frame, the channel is generated according to the distance-dependent large-scale pathloss and
 Rayleigh fading distribution, independent over time and across users, as described in the signal model \eqref{signalmodel}. The SINR is then computed at each SU and PU receiver, and the transmission is declared successful if and only if  
 SINR${>}\mathrm{SINR}_{\mathrm{th}}{=}5\mathrm{dB}$. The SU and PU cell throughputs are then averaged out over the 1000 frames and 100 realizations of the network topology.
\end{itemize}
In addition to IBT, RT, Full-CSI and Uncoordinated schemes mentioned previously, we also evaluate the performance of the consensus-based scheme \cite{li:TVT10}.
We set the degree of each node (cell head) to be $d=5$, based on which we generate a connected graph \cite{Bhuiyan}.
This scheme was originally designed for a single PU cell system without temporal dynamics in the PU spectrum occupancy, 
and therefore it is not optimized to our model, with multiple cells and temporal dynamics of spectrum occupancy in each cell. We argue that a consensus-based scheme, such as \cite{li:TVT10},
 is not well suited to capture the spatial distribution of interference,
nor the temporal dynamics,
 due to the averaging process of consensus in both the spatial and temporal dimensions. Instead, our scheme allows each SU to estimate accurately the state of nearer cells, to which interference will be stronger, and to track more efficiently their temporal dynamics.
Our numerical evaluation in Fig.  \ref{fig:simresreal} confirms this observation: 
the consensus strategy performs poorly, with performance close to the "Uncoordinated" scheme.
On the other hand, the performance of IBT is very close to that of  Full-NSI and significantly outperforms the "Uncoordinated" scheme. This evaluation confirms that, despite the approximation introduced 
in the INR matrix $\boldsymbol{\Phi}{\in} \mathbb{R}^{N_C \times N_C}$, our multi-scale spectrum estimation positively informs network control.
 
\begin{comment} 
 \begin{figure}[t]
\centering  
\includegraphics[width=.4\linewidth,trim = 2mm 0mm 13mm 8mm,clip=true]{vscost}
\caption{
Impact of aggregation cost on the
SU cell throughput, with $0$dB maximum constraint on the
 average INR caused to PUs. $1$ blockage; no delay.}
\label{fig:simres3}
\end{figure}

 \begin{figure}[t]
\centering  
\includegraphics[width=.7\linewidth,trim = 2mm 0mm 13mm 8mm,clip=true]{real}
\caption{
\nm{This figure is based on a simulation with more realistic pathloss. Comparison with consensus scheme.}}
\label{fig:simresreal}
\end{figure}
\end{comment}

\section{Conclusions}
 \label{conclu}
In this paper, we have proposed a multi-scale approach to spectrum sensing in cognitive cellular networks. 
To reduce the cost of acquisition of NSI, 
we have proposed a hierarchical scheme to obtain aggregate 
state information at multiple scales, at each cell.
We have studied analytically the performance 
 of the aggregation scheme in terms of the trade-off among the SU cell throughput, the interference
generated by their activity to PUs, and the mutual interference of SUs.
We have accounted for aggregation delays, local estimation errors, as well as the cost of aggregation.
We have proposed an agglomerative clustering
 algorithm to find a multi-scale aggregation tree, matched to the structure of interference.
We have shown that our proposed design achieves performance close to that with full NSI, using only one-third of the cost of  exchange of spectrum estimates over the network.

\section*{Appendix A: Proof of Lemma \ref{Lem:smt}}
\label{proofofLem:smt}
 \begin{proof}We prove it by induction.
 At level-$1$,
  (\ref{sdfhb}) holds
by definition, see (\ref{smtp1}).
Now, let $L{>}1$ and assume (\ref{sdfhb})
holds at level-$(L{-}1)$. 
The induction hypothesis in (\ref{aggregationL}) implies
%Using (\ref{aggregationL}) and the induction hypothesis we obtain
\begin{align}
S_{m,t}^{(L)}
=
\sum_{k\in\mathcal H_{m}^{(L-1)}}
\sum_{j\in\mathcal C_{k}^{(L-1)}}\hat b_{j,t-\delta_j^{(L-1)}-\Delta_{k}^{(L)}}.
\end{align}
Then, using (\ref{deltaupdate}) we obtain
\begin{align}
S_{m,t}^{(L)}
=
\sum_{k\in\mathcal H_{m}^{(L-1)}}
\sum_{j\in\mathcal C_{k}^{(L-1)}}\hat b_{j,t-\delta_j^{(L)}}
=
\sum_{j\in\mathcal C_{m}^{(L)}}
\hat b_{j,t-\delta_j^{(L)}},
\end{align}
where the last step follows from (\ref{recC}).
The induction step, hence the lemma, are thus proved.
\end{proof}
\section*{Appendix B: Proof of Theorem \ref{thm1}}
\label{proofofthm1}
\begin{proof}
Let $t\geq 0$. Eq. (\ref{eq1}) follows from the fact that 
$\sigma_{i,t}^{(L)}$ is independent of
$b_{j,\tau},\forall \tau\leq t$ for $j\notin\mathcal D_i^{(L)}$,
and from the fact that $(b_{j,\tau},M_{j,\tau},\xi_{j,\tau}),\tau\leq t$
are independent across cells.

We now prove (\ref{eq2}) for a given set
$\mathcal D_i^{(L)}$ and h-distance $L$.
With a slight abuse of notation,
"$\forall j$" should be 
intended as "$\forall j{\in}\mathcal D_i^{(L)}$",
and "$\sum_j$" as "$\sum_{j\in\mathcal D_i^{(L)}}$".
Using \eqref{sigmadel},
\begin{align}
&\mathbb P\Bigr(
b_{j,t}=b_j,\forall j\Bigr|
\sigma_{i,\tau}^{(L)}=o_{\tau}^{(L)},\forall \tau \leq t\Bigr)
%\\&
\nonumber\\&
=
\mathbb P\Bigr(
b_{j,t}=b_j,\forall j\Bigr|
\sum_{j}\hat b_{j,\tau-\delta_j^{(L)}}=o_{\tau}^{(L)},\forall \tau\leq t\Bigr).
\nonumber
\end{align}
%where $\propto$ denotes proportionality up to a normalization factor independent of $\{b_j,\forall j\}$.
We can rewrite it as the marginal with respect to
$b_{j,t-\delta_j^{(L)}},\forall j$
and $\sum_{j}b_{j,t-\delta_j^{(L)}}=x$, yielding
\begin{align}
&\mathbb P\left(
b_{j,t}=b_j,\forall j\left|
\sigma_{i,\tau}^{(L)}=o_{\tau}^{(L)},\forall \tau \leq t\right.\right)
\label{4}
\\&
=
\sum_{x=0}^{|\mathcal D_i^{(L)}|}
\nonumber
\sum_{(\tilde b_{j,t-\delta_j^{(L)}})_{\forall j}}
\mathbb P\Bigr(
b_{j,t}=b_j,\forall j\Bigr|
b_{j,t-\delta_j^{(L)}}{=}
\tilde b_{j,t-\delta_j^{(L)}},\forall j,
\nonumber\\&\quad\nonumber
\quad
\sum_{j}b_{j,t-\delta_j^{(L)}}{=}x,
\sum_{j}\hat b_{j,\tau-\delta_j^{(L)}}=o_{\tau}^{(L)},\forall \tau\leq t\Bigr)
\tag{D;E}
\label{DE}
\\&\times
\nonumber
\mathbb P\Bigr(
b_{j,t-\delta_j^{(L)}}=\tilde b_{j,t-\delta_j^{(L)}},\forall j
\Bigr|
\sum_{j}b_{j,t-\delta_j^{(L)}}=x,
\nonumber\\&\qquad\nonumber
\quad
\sum_{j}\hat b_{j,\tau-\delta_j^{(L)}}=o_{\tau}^{(L)},\forall \tau\leq t\Bigr)
\tag{B;C}
\label{BC}
\\&{\times}
\mathbb P\Bigr(
\sum_{j}b_{j,t-\delta_j^{(L)}}{=}x\Bigr|
\sum_{j}\hat b_{j,\tau-\delta_j^{(L)}}{=}o_{\tau}^{(L)},\forall \tau{\leq}t\Bigr).
\tag{A}
\end{align}
Using the fact that $\{b_{j,t}\}$ is Markov and i.i.d. across cells, 
for the term \eqref{DE} we obtain
\begin{align}
&\mathbb P\Bigr(
b_{j,t}{=}b_j,\forall j\Bigr|
b_{j,t-\delta_j^{(L)}}{=}\tilde b_{j,t-\delta_j^{(L)}},\forall j,
\nonumber\\&\qquad\nonumber
\sum_{j}b_{j,t-\delta_j^{(L)}}=x,\sum_{j}\hat b_{j,\tau-\delta_j^{(L)}}=o_{\tau}^{(L)},\forall \tau\leq t\Bigr)
%\end{multline*}
%\begin{align}
\\&=
\prod_j\mathbb P\left(
b_{j,t}=b_j\left|
b_{j,t-\delta_j^{(L)}}=\tilde b_{j,t-\delta_j^{(L)}}\right.\right),
\label{3}
\end{align}
since $b_{j,t}$ is independent of all other quantities 
given $b_{j,t-\delta_j^{(L)}}$.
In particular,
the probability term in (\ref{3})
is the $\delta_j^{(L)}$ steps transition probability of the Markov chain $\{b_{j,\tau},\forall \tau\}$, i.e.,
\begin{align}
%\nonumber
\label{2}
&\mathbb P\left(
b_{j,t}=b_j
\left|
b_{j,t-\delta_j^{(L)}}=\tilde b_{j}
\right.\right)
\\&\nonumber
{=}
\left[\pi_B{+}\mu^{\delta_j^{(L)}}\left(\tilde b_{j}{-}\pi_B\right)\right]^{b_j}
%\\&\quad\times
\left[1{-}\pi_B{-}\mu^{\delta_j^{(L)}}\left(\tilde b_{j}{-}\pi_B\right)\right]^{1{-}b_j}.
\end{align}
which is equivalent to the terms $D $ and $E$ in \eqref{eq2}.
Next, letting $\mathbf b_t^{(\delta)}=(
b_{j,t-\delta_j^{(L)}})_{\forall j}$, we show that the term \eqref{BC} is equivalent to
$B$ and $C$ in (\ref{eq2}). In fact,
\begin{align}
&\mathbb P\Bigr(
\mathbf b_t^{(\delta)}
{=}
\tilde{\mathbf b}
\Bigr|
\sum_j b_{j,t}^{(\delta)}{=}x,
\sum_{j}\hat b_{j,\tau-\delta_j^{(L)}}{=}o_{\tau}^{(L)},\forall \tau\leq t\Bigr)
\nonumber\\&
%\\&
=\chi\Bigr(\sum_j
\tilde b_{j}=x
\Bigr)
\frac{x!(|\mathcal D_i^{(L)}|-x)!}{|\mathcal D_i^{(L)}|!}.
\label{1}
\end{align}
We obtain (\ref{eq2}) by substituting (\ref{3})-(\ref{1}) into (\ref{4}).

To see (\ref{1}), 
first note that, if 
$\sum_j
\tilde b_{j}\neq x$,
then (\ref{1}) must be zero, since we are conditioning on 
$\sum_j b_{j,t}^{(\delta)}{=}x$.
Thus, we focus on the case 
$\sum_j\tilde b_{j}=x$.
Let $s_j=(M_{j,\tau},{\xi}_{j,\tau})_{\forall j,
-\delta_j^{(L)}\leq\tau\leq t-\delta_j^{(L)}}$
and $\tilde s_j$ be a specific realization of the estimation process.
From the expression of the local estimator, we note that 
$\hat b_{j,\tau}$ is a function of 
$(\hat b_{j,\tau-1},M_{j,\tau},\xi_{j,\tau})$.
Then, by induction, 
$\hat b_{j,\tau}$ is a function of 
$(M_{j,\tau^\prime},\xi_{j,\tau^\prime})_{-\delta_j^{(L)}\leq\tau^\prime\leq\tau}$
(and thus of $s_j$), denoted as
\\
\centerline{$\hat b_{j,\tau}=
g_{\tau+\delta_j^{(L)}}(s_j).$}
\\
Note that the subscript  $\tau+\delta_j^{(L)}$
signifies that  the first 
$\tau+\delta_j^{(L)}+1$ samples of 
$(M_{j,\tau^\prime},\xi_{j,\tau^\prime})$
are used to compute $\hat b_{j,\tau}$,
since $-\delta_j^{(L)}\leq\tau^\prime\leq\tau$.
Importantly, $\hat b_{j,\tau}$ depends on the cell index $j$ only through 
$\delta_j^{(L)}$ and $s_j$, so that
\\
\centerline{$\sum_j\hat b_{j,\tau-\delta_j^{(L)}}
=
\sum_jg_{\tau}(s_j).$}
\\
Let $\mathcal {BS}$ be the set of 
tuples $(\mathbf b,\mathbf s)$
such that 
\\\centerline{
$\sum_j b_{j}=x
,\ 
\sum_jg_{\tau}(s_j)=o_{\tau}^{(L)},\forall 0\leq\tau\leq t.$}
\\
Using this definition, we write
the left hand side of
(\ref{1}) as
\begin{align}
%\nonumber
&\mathbb P\Bigr(
\mathbf b_{t}^{(\delta)}{=}\tilde{\mathbf b}
\Bigr|
\sum_{j}b_{j}=x,
\sum_{j}\hat b_{j,\tau-\delta_j^{(L)}}=o_{\tau}^{(L)},\forall \tau\leq t\Bigr)
\nonumber\\&
%\\&
=\mathbb P\left(
\mathbf b_{t}^{(\delta)}=\tilde{\mathbf b}
\left|
(\mathbf b_{t}^{(\delta)},\mathbf s)\in\mathcal {BS}
\right.\right).
\label{1_2}
\end{align}
Consider a permutation $\mathcal P:\mathcal D_i^{(L)}\mapsto\mathcal D_i^{(L)}$ of the elements in the set $\mathcal D_i^{(L)}$.
Thus,
\begin{align}
\left\{
\begin{array}{l}
\sum_jg_{\tau}(s_{\mathcal P(j)})
=
\sum_jg_{\tau}(s_j),
\\
\sum_j b_{\mathcal P(j)}
=
\sum_j b_{j}=x.
\end{array}
\right.
\end{align}
By definition of $\mathcal {BS}$,
if $(\mathbf b,\mathbf s){\in}\mathcal {BS}$ then, under any permutation, $(\mathbf b_{\mathcal P},\mathbf s_{\mathcal P})
{\in}\mathcal {BS}$,
where 
$\mathbf b_{\mathcal P}=(b_{\mathcal P(j)})_{\forall j}$
and
$\mathbf s_{\mathcal P}=(s_{\mathcal P(j)})_{\forall j}$. We can thus partition
$\mathcal {BS}$ into 
$|\mathcal U|$ sets,
$\mathcal {BS}_{u},u\in\mathcal U$,
where $\mathcal U$ is a set of indexes,
such that
$ \mathcal {BS}_u$ contains all and only  the permutations of its elements, that is
\begin{align}
\left\{
\begin{array}{l}
(\mathbf b,\mathbf s)\in\mathcal {BS}_u\Leftrightarrow
(\mathbf b_{\mathcal P},\mathbf s_{\mathcal P})
\in\mathcal {BS}_u,\ \forall \mathcal P,
\\
(\mathbf b^{(1)},\mathbf s^{(1)})
\in\mathcal {BS}_u,
(\mathbf b^{(1)},\mathbf s^{(1)})\neq
(\mathbf b^{(2)}_{\mathcal P},\mathbf s^{(2)}_{\mathcal P}),\forall \mathcal P\\\qquad\qquad\Rightarrow
(\mathbf b^{(2)}_{\mathcal P},\mathbf s^{(2)}_{\mathcal P})\notin\mathcal {BS}_u,
\\
\cup_{u\in\mathcal U}\mathcal {BS}_u\equiv\mathcal{BS},
\ 
\mathcal {BS}_{u_1}\cap\mathcal {BS}_{u_2}\equiv\emptyset,\forall u_1\neq u_2.
\end{array}
\right.
\end{align}
By marginalizing with respect to the realization of the sequence $(\mathbf b,\mathbf s)$,
 we then obtain
 \begin{align}
\label{1_3}
&
\mathbb P\left(
\mathbf b_{t}^{(\delta)}{=}\tilde{\mathbf b}
\left|
(\mathbf b_{t}^{(\delta)},\mathbf s)\in\mathcal {BS}
\right.\right)
%\\&
\nonumber\\&
=
\sum_{u\in\mathcal U}
\sum_{(\bar{\mathbf b},\bar{\mathbf s})\in\mathcal {BS}_u}
\chi(\tilde{\mathbf b}{=}\bar{\mathbf b})
\mathbb P\bigl(
(\mathbf b_{t}^{(\delta)},\mathbf s){=}
(\bar{\mathbf b},\bar{\mathbf s})
\left|
(\mathbf b_{t}^{(\delta)},\mathbf s)\in\mathcal {BS}_u
\right.\bigr)
\nonumber
\\&\qquad\times\mathbb P\bigl(
(\mathbf b_{t}^{(\delta)},\mathbf s)\in\mathcal {BS}_u
\left|
(\mathbf b_{t}^{(\delta)},\mathbf s)\in\mathcal {BS}
\right.\bigr).
%\nonumber
\end{align}
Let $(\mathbf b^{(1)},\mathbf s^{(1)}){\in}\mathcal {BS}_u$ and 
$(\mathbf b^{(2)},\mathbf s^{(2)}){\in}\mathcal {BS}_u$. By definition of
 $\mathcal {BS}_u$, we have that  $(\mathbf b^{(2)},\mathbf s^{(2)})
{=}(\mathbf b_{\mathcal P}^{(1)},\mathbf s_{\mathcal P}^{(1)})
$ under some permutation 
$\mathcal P$. Since $\{(b_{j,\tau},M_{j,\tau},\xi_{j,\tau}),-\delta_j^{(L)}\leq\tau\leq t-\delta_j^{(L)}\}$
is stationary over time and i.i.d. across cells,
by permuting this sequence 
across cells, we obtain a sequence 
with the same probability of occurrence;
in other words,
\begin{align}
&\mathbb P\left(
(\mathbf b_{t}^{(\delta)},\mathbf s)
=
(\mathbf b^{(1)},\mathbf s^{(1)})
\left|
(\mathbf b_{t}^{(\delta)},\mathbf s)
\in\mathcal {BS}_u
\right.\right)
\nonumber
\\&\nonumber
%\\&\nonumber
=\mathbb P\left(
(\mathbf b_{t}^{(\delta)},\mathbf s)
=
(\mathbf b^{(2)},\mathbf s^{(2)})
\left|
(\mathbf b_{t}^{(\delta)},\mathbf s)
\in\mathcal {BS}_u
\right.\right)
\\&
=
\mathbb P\left(
(\mathbf b_{t}^{(\delta)},\mathbf s)
=
(\mathbf b_{\mathcal P}^{(1)},\mathbf s_{\mathcal P}^{(1)})
\left|
(\mathbf b_{t}^{(\delta)},\mathbf s)
\in\mathcal {BS}_u
\right.\right),\ \forall \mathcal P.
\end{align}
Hence, $(\mathbf b_{t}^{(\delta)},\mathbf s)
$ has uniform distribution
over the set $\mathcal {BS}_u$, and
we must have 
\\
\centerline{$\mathbb P\left(
(\mathbf b_{t}^{(\delta)},\mathbf s)
{=}
(\mathbf b^{(1)},\mathbf s^{(1)})
\left|
(\mathbf b_{t}^{(\delta)},\mathbf s)
{\in}\mathcal {BS}_u
\right.\right)
{=}\frac{1}{|\mathcal {BS}_u|}
{=}\frac{1}{|\mathcal D_i^{(L)}|!},$}\\
corresponding to all possible permutations.
Substituting in (\ref{1_3}), we then obtain
 \begin{align}
&
\mathbb P\left(
\mathbf b_{t}^{(\delta)}=\tilde{\mathbf b}
\left|
(\mathbf b_{t}^{(\delta)},\mathbf s)\in\mathcal {BS}
\right.\right)
{=}
\frac{1}{|\mathcal D_i^{(L)}|!}
\sum_{u\in\mathcal U}
\sum_{(\bar{\mathbf b},\bar{\mathbf s})\in\mathcal {BS}_u}
\nonumber
\\&
\qquad\chi\left(\tilde{\mathbf b}{=}\bar{\mathbf b}\right)
%\nonumber\\&\qquad\times
\mathbb P\bigl(
(\mathbf b_{t}^{(\delta)},\mathbf s){\in}\mathcal {BS}_u
\left|
(\mathbf b_{t}^{(\delta)},\mathbf s){\in}\mathcal {BS}
\right.\bigr).
\label{1_4}
\end{align}
Since there are exactly
$x!(|\mathcal D_i^{(L)}|{-}x)!$
combinations of $(\mathbf b_{t}^{(\delta)},\mathbf s)$ within $\mathcal {BS}_u$ such that 
$\mathbf b_{t}^{(\delta)}{=}\tilde{\mathbf b}$
(since 
$\sum_j\tilde b_{j}{=}x$ by assumption), we obtain
\begin{align}
\sum_{(\bar{\mathbf b},\bar{\mathbf s})\in\mathcal {BS}_u}
\chi(\tilde{\mathbf b}=\bar{\mathbf b})
=
x!(|\mathcal D_i^{(L)}|-x)!.
\end{align}
Substituting in (\ref{1_4}), we finally obtain
\begin{align}
\nonumber
&
\mathbb P\left(
\mathbf b_{t}^{(\delta)}=\tilde{\mathbf b}
\left|
(\mathbf b_{t}^{(\delta)},\mathbf s)\in\mathcal {BS}
\right.\right)
%\\&
\\&\nonumber
=
\frac{
x!(|\mathcal D_i^{(L)}|-x)!
}{|\mathcal D_i^{(L)}|!}
\sum_{u\in\mathcal U}
\mathbb P\left(
(\mathbf b_{t}^{(\delta)},\mathbf s)\in\mathcal {BS}_u
\left|
(\mathbf b_{t}^{(\delta)},\mathbf s)\in\mathcal {BS}
\right.\right)
%\nonumber\\&
\nonumber\\&=
\frac{
x!(|\mathcal D_i^{(L)}|-x)!
}{|\mathcal D_i^{(L)}|!},
\end{align}
which proves (\ref{1}) when $\sum_j\tilde b_{j}{=}x$.
Eq. (\ref{eq2}) is thus proved.

To conclude the proof of Theorem \ref{thm1}, we prove (\ref{eq2_2}).
We rewrite the left hand side of (\ref{eq2_2}) as
\begin{align}
&\Theta\triangleq
\mathbb E\Bigr(
\sum_{j}b_{j,t}^{(\delta)}
\Bigr|
\sum_jg_{\tau}(s_j)=o_{\tau}^{(L)}, 0{\leq}\tau{\leq}t
\Bigr).
\end{align}
Now, assume a genie-aided case
which directly observes the sequence $\mathbf s$,
rather than the aggregates 
$\sum_jg_{\tau}(s_j),\forall 0\leq\tau\leq t$.
Using the notation
of the previous part of the proof, 
 let $\tilde{\mathbf s}$
be a specific realization such that
$\sum_jg_{\tau}(\tilde s_j){=}o_{\tau}^{(L)},\forall 0{\leq}\tau{\leq}t$.
In the genie aided case, by the linearity of expectation we obtain
\begin{align*}
&\mathbb E\Bigr(
\sum_{j}b_{j,t}^{(\delta)}\Bigr|
\mathbf s=\tilde{\mathbf s}
\Bigr)
=
\sum_{j}
\mathbb P\left(\left.
b_{j,t}^{(\delta)}{=}1\right|
\mathbf s{=}\tilde{\mathbf s}
\right).
\end{align*}
Since $b_{j,t}^{(\delta)}$ is
statistically independent of $s_{j^\prime}$ for
$j^\prime\neq j$ given $s_{j}$, 
by definition of $s_j$ it follows that
\begin{align*}
&\mathbb E\Bigr(
\sum_{j}b_{j,t}^{(\delta)}\Bigr|
\mathbf s=\tilde{\mathbf s}
\Bigr)
=
\nonumber
\sum_{j}
\mathbb P\left(\left.
b_{j,t}^{(\delta)}=1\right|
s_{j}=\tilde s_{j}
\right)
\\&
{=}
\sum_{j}
\mathbb P\left(
b_{j,t}^{(\delta)}{=}1\left|
(M_{j,\tau},\xi_{j\tau})
{=}(\tilde M_{j,\tau},\tilde{\xi}_{j\tau}),
\vphantom{\tilde{\xi}_{j\tau}}
{-}\delta_j^{(L)}{\leq}\tau{\leq} t{-}\delta_j^{(L)}
\right.
\right)
\nonumber\\&
=
\sum_{j}
\hat b_{j,t-\delta_j^{(L)}}
\triangleq
g_t(\mathbf s).
\nonumber
\end{align*}
Thus, $g_t(\mathbf s)$ is sufficient to 
compute the posterior expectation of 
$\sum_{j}b_{j,t}^{(\delta)}$ in the genie-aided case.
Since $g_t(\mathbf s)$ is also available in the non-genie-aided case, it must be the case that $\Theta=g_t(\mathbf s)$ as well, yielding (\ref{eq2_2}) via \eqref{sigmadel}.
The theorem is thus proved.
\end{proof}
 \section*{Appendix C: Proof of Lemma \ref{lem2}}
\label{proofoflem2}
\begin{proof}
Let  $0{\leq}L{\leq}D$ and $j\in\mathcal D_i^{(L)}$. 
Using (\ref{eq1}) we obtain
\begin{align}
\label{x1}
&\mathbb P(b_{j,t}=1|\pi_{i,t})
=
\sum_{\mathbf b}\chi(b_j=1)\pi_{i,t}(\mathbf b)
\\&
{=}\nonumber
\sum_{b_{j^\prime},\forall j^\prime\in\mathcal D_i^{(L)}}
\chi(b_{j}{=}1)
\mathbb P\left(\left .b_{j^\prime,t}{=}b_{j^\prime},\forall j^\prime{\in}\mathcal D_i^{(L)}\right|
\sigma_{i,\tau}^{(L)}{=}o_{\tau}^{(L)},\forall{\tau}{\leq}t\right).
\end{align}
Since we are considering only the cells in the set
$\mathcal D_i^{(L)}$,
with a slight abuse of notation,
"$\forall j$" should be intended as "$\forall j{\in}\mathcal D_i^{(L)}$";
"$\sum_j$" as "$\sum_{j\in\mathcal D_i^{(L)}}$";
and vectors are restricted to their indices in $\mathcal D_i^{(L)}$.
Let $\mathbf b_t^{(\delta)}=(b_{j,t-\delta_j^{(L)}})_{\forall j}$.
Using (\ref{eq2}), we can rewrite (\ref{x1}) as
\begin{align}
&\mathbb P(b_{j,t}=1|\pi_{i,t})
=
\sum_{\tilde{\mathbf b}}
\mathbb P(b_{j,t}=1|b_{j,t}^{(\delta)}=\tilde b_j)
\nonumber\\&\qquad\times
\mathbb P\left(\left .
\mathbf b_t^{(\delta)}=\tilde{\mathbf b}
\right|\sigma_{i,\tau}^{(L)}=o_{\tau}^{(L)},\forall \tau \leq t\right),
\end{align}
where
\\\centerline{
$\mathbb P(b_{j,t}=1|b_{j,t}^{(\delta)}=\tilde b_j)
=\pi_B+\mu^{\delta_j^{(L)}}\left(\tilde b_{j}-\pi_B\right).$}\\
($\delta_j^{(L)}$ steps transition probability to
$b_{j,t}=1$)
and
\begin{align}
\label{y2}
&\mathbb P\Bigr(
\mathbf b_t^{(\delta)}{=}\tilde{\mathbf b}
\Bigr|\sigma_{i,\tau}^{(L)}=o_{\tau}^{(L)},\forall \tau \leq t\Bigr)
%\\&\nonumber
\nonumber\\&
{=}
\sum_{x=0}^{|\mathcal D_i^{(L)}|}
\mathbb P\Bigr(
\sum_{j}b_{j,t}^{(\delta)}
{=}x\Bigr|
\sigma_{i,\tau}^{(L)}=o_{\tau}^{(L)},\forall \tau \leq t\Bigr)
\nonumber\\&\qquad\times
\frac{x!(|\mathcal D_i^{(L)}|-x)!}{|\mathcal D_i^{(L)}|!}
\chi\Bigr(\sum_{j}\tilde b_{j}=x\Bigr).
\end{align}
Thus, we obtain
\begin{align}
\label{w2}
&\mathbb P(b_{j,t}=1|\pi_{i,t})
=
\pi_B(1-\mu^{\delta_j^{(L)}})
%\nonumber\\&\qquad
\nonumber\\&\quad
+\mu^{\delta_j^{(L)}}
\sum_{\tilde{\mathbf b}}\tilde b_{j}
\mathbb P\left(\left .
\mathbf b_t^{(\delta)}=\tilde{\mathbf b}
\right|\sigma_{i,\tau}^{(L)}=o_{\tau}^{(L)},\forall \tau \leq t\right).
\end{align}
Now, using (\ref{y2}) we obtain
\begin{align}
\label{w1}
&\sum_{\tilde{\mathbf b}}\tilde b_{j}
\mathbb P\Bigr(
\mathbf b_t^{(\delta)}=\tilde{\mathbf b}
\Bigr|\sigma_{i,\tau}^{(L)}=o_{\tau}^{(L)},\forall \tau \leq t\Bigr)\\&\nonumber
=
\sum_{x=1}^{|\mathcal D_i^{(L)}|}
\mathbb P\Bigr(
\sum_{j^\prime}b_{j^\prime}
=x\Bigr|
\sigma_{i,\tau}^{(L)}=o_{\tau}^{(L)},\forall \tau \leq t\Bigr)
%\\&\times
\nonumber\\&\qquad\times
\frac{x!(|\mathcal D_i^{(L)}|-x)!}{|\mathcal D_i^{(L)}|!}
\sum_{\tilde{\mathbf b}}\tilde b_{j}
\chi\Bigr(\sum_{j^\prime}\tilde b_{j^\prime}=x\Bigr).
\nonumber
\end{align}
Note that the sum over $x$ starts from $x=1$ instead of $x=0$. In fact, 
if $x=0$, then $\tilde{\mathbf b}=\mathbf 0$ and $\tilde b_j=0$,
which does not contribute to
(\ref{w1}).
Finally, since there are  $|\mathcal D_i^{(L)}|-1$ over $x-1$
possible combinations of vectors $\tilde{\mathbf b}\in\{0,1\}^{|\mathcal D_i^{(L)}|}$
such that
$\tilde b_{j}=1$ and $\sum_{j^\prime}\tilde b_{j^\prime}=x$,
 we obtain
\\\centerline{
$ \frac{x!(|\mathcal D_i^{(L)}|-x)!}{|\mathcal D_i^{(L)}|!}
\sum_{\tilde{\mathbf b}}\tilde b_{j}
\chi\Bigr(\sum_{j^\prime}\tilde b_{j^\prime}=x\Bigr)
=
\frac{x}{|\mathcal D_i^{(L)}|},
$}\\
 hence
 \begin{align}
 \label{w3}
&\sum_{\tilde{\mathbf b}}\tilde b_{j}
\mathbb P\Bigr(
\mathbf b_t^{(\delta)}=\tilde{\mathbf b}
\Bigr|\sigma_{i,\tau}^{(L)}=o_{\tau}^{(L)},\forall \tau \leq t\Bigr)
\\&\nonumber
{=}
\frac{1}{|\mathcal D_i^{(L)}|}
%\\&\nonumber\times
\sum_{x{=}0}^{|\mathcal D_i^{(L)}|}x
\mathbb P\Bigr(
\sum_{j^\prime}b_{j^\prime}
{=}x\Bigr|
\sigma_{i,\tau}^{(L)}{=}o_{\tau}^{(L)},\forall\tau{\leq}t\Bigr)
{=}
\frac{o_{t}^{(L)}}{|\mathcal D_i^{(L)}|},
\end{align}
where in the last step we used 
(\ref{eq2_2}).
The lemma is thus proved by substituting
(\ref{w3}) into (\ref{w2}).
\end{proof}
%\IEEEtriggeratref{3}
\bibliographystyle{IEEEtran}
\bibliography{IEEEabrv,bibliography} 

\begin{thebibliography}{10}
\providecommand{\url}[1]{#1}
\csname url@samestyle\endcsname
\providecommand{\newblock}{\relax}
\providecommand{\bibinfo}[2]{#2}
\providecommand{\BIBentrySTDinterwordspacing}{\spaceskip=0pt\relax}
\providecommand{\BIBentryALTinterwordstretchfactor}{4}
\providecommand{\BIBentryALTinterwordspacing}{\spaceskip=\fontdimen2\font plus
\BIBentryALTinterwordstretchfactor\fontdimen3\font minus
  \fontdimen4\font\relax}
\providecommand{\BIBforeignlanguage}[2]{{%
\expandafter\ifx\csname l@#1\endcsname\relax
\typeout{** WARNING: IEEEtran.bst: No hyphenation pattern has been}%
\typeout{** loaded for the language `#1'. Using the pattern for}%
\typeout{** the default language instead.}%
\else
\language=\csname l@#1\endcsname
\fi
#2}}
\providecommand{\BIBdecl}{\relax}
\BIBdecl

\bibitem{MichelusiGCOM}
N.~Michelusi, M.~Nokleby, U.~Mitra, and R.~Calderbank, ``{Dynamic Spectrum
  Estimation with Minimal Overhead via Multiscale Information Exchange},'' in
  \emph{IEEE Global Communications Conference (GLOBECOM)}, Dec 2015, pp. 1--6.

\bibitem{MicheICC}
------, ``Multi-scale spectrum sensing in small-cell mm-wave cognitive wireless
  networks,'' in \emph{2017 IEEE International Conference on Communications
  (ICC)}, May 2017, pp. 1--6.

\bibitem{MicheAsilomar}
------, ``{Multi-scale spectrum sensing in millimeter wave cognitive
  networks},'' in \emph{2017 51st Asilomar Conference on Signals, Systems, and
  Computers}, Oct 2017, pp. 1640--1644.

\bibitem{CISCO}
\BIBentryALTinterwordspacing
CISCO, ``{Cisco Visual Networking Index: Global Mobile Data Traffic Forecast
  Update, 2015-2020 White Paper},'' Tech. Rep. [Online]. Available:
  \url{http://www.cisco.com/c/en/us/solutions/collateral/service-provider/visual-networking-index-vni/mobile-white-paper-c11-520862.html}
\BIBentrySTDinterwordspacing

\bibitem{pcast}
\BIBentryALTinterwordspacing
``{Realizing the Full Potential of Government-Held Spectrum to Spur Economic
  Growth},'' Tech. Rep., July 2012, report to the president. [Online].
  Available:
  \url{http://www.whitehouse.gov/sites/default/files/microsites/ostp/
  pcast_spectrum_report_final_july_20_2012.pdf}
\BIBentrySTDinterwordspacing

\bibitem{Mitola}
J.~Mitola and G.~Maguire, ``{Cognitive radio: making software radios more
  personal},'' \emph{IEEE Personal Communications}, vol.~6, no.~4, pp. 13--18,
  Aug. 1999.

\bibitem{Peha}
J.~Peha, ``{Sharing Spectrum Through Spectrum Policy Reform and Cognitive
  Radio},'' \emph{Proceedings of the IEEE}, vol.~97, no.~4, pp. 708--719, Apr.
  2009.

\bibitem{Wu2012}
Q.~Wu, G.~Ding, J.~Wang, X.~Li, and Y.~Huang, ``Consensus-based decentralized
  clustering for cooperative spectrum sensing in cognitive radio networks,''
  \emph{Chinese Science Bulletin}, vol.~57, no.~28, pp. 3677--3683, Oct 2012.

\bibitem{Letaief}
W.~Zhang, R.~K. Mallik, and K.~B. Letaief, ``{Optimization of cooperative
  spectrum sensing with energy detection in cognitive radio networks},''
  \emph{IEEE Transactions on Wireless Communications}, vol.~8, no.~12, pp.
  5761--5766, Dec. 2009.

\bibitem{Ding}
G.~Ding, J.~Wang, Q.~Wu, L.~Zhang, Y.~Zou, Y.~D. Yao, and Y.~Chen, ``{Robust
  Spectrum Sensing With Crowd Sensors},'' \emph{IEEE Transactions on
  Communications}, vol.~62, no.~9, pp. 3129--3143, Sept 2014.

\bibitem{Ejaz}
W.~Ejaz, G.~Hattab, N.~Cherif, M.~Ibnkahla, F.~Abdelkefi, and M.~Siala,
  ``{Cooperative Spectrum Sensing With Heterogeneous Devices: Hard Combining
  Versus Soft Combining},'' \emph{IEEE Systems Journal}, vol.~12, no.~1, pp.
  981--992, March 2018.

\bibitem{Goeckel}
D.~L. Goeckel, ``Adaptive coding for time-varying channels using outdated
  fading estimates,'' \emph{IEEE Transactions on Communications}, vol.~47,
  no.~6, pp. 844--855, Jun 1999.

\bibitem{Vasilakos}
D.~Lopez-Perez, X.~Chu, A.~V. Vasilakos, and H.~Claussen, ``On distributed and
  coordinated resource allocation for interference mitigation in
  self-organizing lte networks,'' \emph{IEEE/ACM Transactions on Networking},
  vol.~21, no.~4, pp. 1145--1158, Aug 2013.

\bibitem{nokleby:JSTSP13}
M.~Nokleby, W.~U. Bajwa, A.~R. Calderbank, and B.~Aazhang, ``Toward
  resource-optimal consensus over the wireless medium,'' \emph{{IEEE} Journal
  of Selected Topics in Signal Processing}, vol.~7, no.~2, Apr. 2013.

\bibitem{benezit:IT10}
F.~Benezit, A.~Dimakis, P.~Thiran, and M.~Vetterli, ``Order-optimal consensus
  through randomized path averaging,'' \emph{IEEE Transactions on Information
  Theory}, vol.~56, no.~10, pp. 5150 --5167, oct. 2010.

\bibitem{friedman:01}
J.~Friedman, T.~Hastie, and R.~Tibshirani, \emph{The elements of statistical
  learning}.\hskip 1em plus 0.5em minus 0.4em\relax Springer, 2001, vol.~1.

\bibitem{li:TVT10}
Z.~Li, F.~R. Yu, and M.~Huang, ``A distributed consensus-based cooperative
  spectrum-sensing scheme in cognitive radios,'' \emph{IEEE Transactions on
  Vehicular Technology}, vol.~59, no.~1, pp. 383--393, 2010.

\bibitem{zeng:JSTSP11}
Z.~Fanzi, C.~Li, and Z.~Tian, ``Distributed compressive spectrum sensing in
  cooperative multihop cognitive networks,'' \emph{IEEE Journal of Selected
  Topics in Signal Processing}, vol.~5, no.~1, pp. 37--48, 2011.

\bibitem{Hajihoseini}
A.~Hajihoseini and S.~A. Ghorashi, ``Distributed spectrum sensing for cognitive
  radio sensor networks using diffusion adaptation,'' \emph{IEEE Sensors
  Letters}, vol.~1, no.~5, pp. 1--4, Oct 2017.

\bibitem{Guoru}
Q.~Wu, G.~Ding, J.~Wang, and Y.~Yao, ``{Spatial-Temporal Opportunity Detection
  for Spectrum-Heterogeneous Cognitive Radio Networks: Two-Dimensional
  Sensing},'' \emph{IEEE Transactions on Wireless Communications}, vol.~12,
  no.~2, pp. 516--526, February 2013.

\bibitem{myTCNC}
N.~Michelusi and U.~Mitra, ``{Cross-Layer Estimation and Control for Cognitive
  Radio: Exploiting Sparse Network Dynamics},'' \emph{IEEE Transactions on
  Cognitive Communications and Networking}, vol.~1, no.~1, pp. 128--145, March
  2015.

\bibitem{Bazerque}
J.~A. Bazerque and G.~B. Giannakis, ``Distributed spectrum sensing for
  cognitive radio networks by exploiting sparsity,'' \emph{IEEE Transactions on
  Signal Processing}, vol.~58, no.~3, pp. 1847--1862, March 2010.

\bibitem{MicheTSP1}
N.~Michelusi and U.~Mitra, ``{Cross-Layer Design of Distributed
  Sensing-Estimation With Quality Feedback -- Part I: Optimal Schemes},''
  \emph{IEEE Transactions on Signal Processing}, vol.~63, no.~5, pp.
  1228--1243, March 2015.

\bibitem{Sun}
S.~Sun, T.~S. Rappaport, T.~A. Thomas, A.~Ghosh, H.~C. Nguyen, I.~Z. Kovács,
  I.~Rodriguez, O.~Koymen, and A.~Partyka, ``Investigation of prediction
  accuracy, sensitivity, and parameter stability of large-scale propagation
  path loss models for 5g wireless communications,'' \emph{IEEE Transactions on
  Vehicular Technology}, vol.~65, no.~5, pp. 2843--2860, May 2016.

\bibitem{Bernstein}
D.~S. Bernstein, S.~Zilberstein, and N.~Immerman, ``{The Complexity of
  Decentralized Control of Markov Decision Processes},'' in \emph{Proceedings
  of the Sixteenth Conference on Uncertainty in Artificial Intelligence}, ser.
  UAI'00.\hskip 1em plus 0.5em minus 0.4em\relax San Francisco, CA, USA: Morgan
  Kaufmann Publishers Inc., 2000, pp. 32--37.

\bibitem{Rentel}
C.~H. Rentel and T.~Kunz, ``{A Mutual Network Synchronization Method for
  Wireless Ad Hoc and Sensor Networks},'' \emph{IEEE Transactions on Mobile
  Computing}, vol.~7, no.~5, pp. 633--646, May 2008.

\bibitem{Bhuiyan}
H.~Bhuiyan, M.~Khan, and M.~Marathe, ``{A parallel algorithm for generating a
  random graph with a prescribed degree sequence},'' in \emph{2017 IEEE
  International Conference on Big Data (Big Data)}, Dec 2017, pp. 3312--3321.

\bibitem{bai:TWC2014}
T.~Bai, R.~Vaze, and R.~W. Heath, ``Analysis of blockage effects on urban
  cellular networks,'' \emph{IEEE Transactions on Wireless Communications},
  vol.~13, no.~9, pp. 5070--5083, 2014.

\end{thebibliography}

\begin{IEEEbiography}
    [{\includegraphics[width=1in,height=1.25in,clip,keepaspectratio]{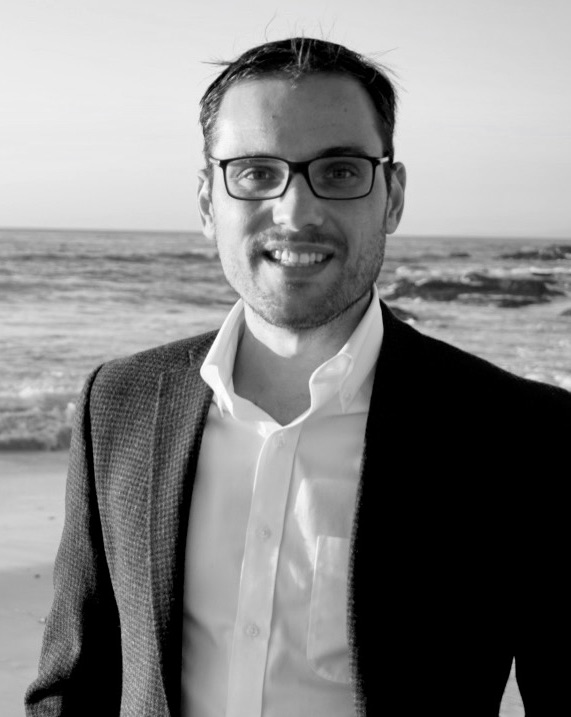}}]{Nicolo Michelusi}(S'09, M'13, SM'18) received the B.Sc. (with honors), M.Sc. (with honors) and Ph.D. degrees from the University of Padova, Italy, in 2006, 2009 and 2013, respectively, and the M.Sc. degree in Telecommunications Engineering from the Technical University of Denmark in 2009, as part of the T.I.M.E. double degree program. He was a post-doctoral research fellow at the Ming-Hsieh Department of Electrical Engineering, University of Southern California, USA, in 2013-2015. He is currently an Assistant Professor at the School of Electrical and Computer Engineering at Purdue University, IN, USA. His research interests lie in the areas of 5G wireless networks, millimeter-wave communications, stochastic optimization, distributed optimization. Dr. Michelusi serves as Associate Editor for the IEEE Transactions on Wireless Communications, and as a reviewer for several IEEE Transactions.
\end{IEEEbiography}

\begin{IEEEbiographynophoto}
{Matthew Nokleby} (S04--M13) received the B.S. (cum laude) and M.S. degrees from Brigham Young University, Provo, UT, in 2006 and 2008, respectively, and the Ph.D. degree from Rice University, Houston, TX, in 2012, all in electrical engineering. From 2012--2015 he was a postdoctoral research associate in the Department of Electrical and Computer Engineering at Duke University, Durham, NC. In 2015 he joined the Department of Electrical and Computer Engineering at Wayne State University as an assistant professor. His research interests span machine learning, signal processing, and information theory, including distributed learning and optimization, sensor networks, and wireless communication. Dr. Nokleby received the Texas Instruments Distinguished Fellowship (2008-2012) and the Best Dissertation Award (2012) from the Department of Electrical and Computer Engineering at Rice University.
\end{IEEEbiographynophoto}

\begin{IEEEbiographynophoto}
{Urbashi Mitra} received the B.S. and the M.S. degrees from the University of California at Berkeley and her Ph.D. from Princeton University.  Dr. Mitra is currently the Gordon S. Marshall Professor in Engineering at the University of Southern California with appointments in Electrical Engineering and Computer Science. She is the inaugural Editor-in-Chief for the IEEE Transactions on Molecular, Biological and Multi-scale Communications. She has been a member of the IEEE Information Theory Society's Board of Governors (2002-2007, 2012-2017), the IEEE Signal Processing Society's Technical Committee on Signal Processing for Communications and Networks (2012-2016), the IEEE Signal Processing Society's Awards Board (2017-2018), and the Vice Chair of the IEEE Communications Society, Communication Theory Working Group (2017-2018). Dr. Mitra is a Fellow of the IEEE.  She is the recipient of: the 2017 IEEE Women in Communications Engineering Technical Achievement Award, a 2015 UK Royal Academy of Engineering Distinguished Visiting Professorship, a 2015 US Fulbright Scholar Award, a 2015-2016 UK Leverhulme Trust Visiting Professorship, IEEE Communications Society Distinguished Lecturer, 2012 Globecom Signal Processing for Communications Symposium Best Paper Award, 2012 US National Academy of Engineering Lillian Gilbreth Lectureship, the 2009 DCOSS Applications \& Systems Best Paper Award, Texas Instruments Visiting Professorship (Fall 2002, Rice University), 2001 Okawa Foundation Award, 2000 Ohio State University's College of Engineering Lumley Award for Research, 1997 Ohio State University's College of Engineering MacQuigg Award for Teaching, and a 1996 National Science Foundation CAREER Award.  She has been an Associate Editor for the following IEEE publications: Transactions on Signal Processing, Transactions on Information Theory, Journal of Oceanic Engineering, and Transactions on Communications.  Dr. Mitra has held visiting appointments at: King's College, London, Imperial College, the Delft University of Technology, Stanford University, Rice University, and the Eurecom Institute. Her research interests are in: wireless communications, communication and sensor networks, biological communication systems, detection and estimation and the interface of communication, sensing and control.
\end{IEEEbiographynophoto}

\begin{IEEEbiographynophoto}
{Robert Calderbank} (M'89, SM'97, F'98) received the BSc degree in 1975 from Warwick University, England, the MSc degree in 1976 from Oxford University, England, and the PhD degree in 1980 from the California Institute of Technology, all in mathematics.

Dr. Calderbank is Professor of Electrical and Computer Engineering at Duke University where he directs the Information Initiative at Duke (iiD). Prior to joining Duke in 2010, Dr. Calderbank was Professor of Electrical Engineering and Mathematics at Princeton University where he directed the Program in Applied and Computational Mathematics. Prior to joining Princeton in 2004, he was Vice President for Research at AT\&T, responsible for directing the first industrial research lab in the world where the primary focus is data at scale. At the start of his career at Bell Labs, innovations by Dr. Calderbank were incorporated in a progression of voiceband modem standards that moved communications practice close to the Shannon limit. Together with Peter Shor and colleagues at AT\&T Labs he developed the mathematical framework for quantum error correction. He is a co-inventor of space-time codes for wireless communication, where correlation of signals across different transmit antennas is the key to reliable transmission.

Dr. Calderbank served as Editor in Chief of the IEEE TRANSACTIONS ON INFORMATION THEORY from 1995 to 1998, and as Associate Editor for Coding Techniques from 1986 to 1989. He was a member of the Board of Governors of the IEEE Information Theory Society from 1991 to 1996 and from 2006 to 2008. Dr. Calderbank was honored by the IEEE Information Theory Prize Paper Award in 1995 for his work on the Z4 linearity of Kerdock and Preparata Codes (joint with A.R. Hammons Jr., P.V. Kumar, N.J.A. Sloane, and P. Sole), and again in 1999 for the invention of space-time codes (joint with V. Tarokh and N. Seshadri). He has received the 2006 IEEE Donald G. Fink Prize Paper Award, the IEEE Millennium Medal, the 2013 IEEE Richard W. Hamming Medal, and the 2015 Shannon Award. He was elected to the US National Academy of Engineering in 2005.
\end{IEEEbiographynophoto}

\end{document}